\documentclass[10pt]{amsart}
\usepackage{a4wide,color}
\usepackage{mathrsfs}
\usepackage{amsbsy}
\usepackage{amstext}
\usepackage{amsthm}
\usepackage{amssymb}
\usepackage[active]{srcltx}
\usepackage{ulem}
\usepackage{enumerate}
\usepackage{cancel}
\makeatletter
\numberwithin{equation}{section}
\numberwithin{figure}{section}
\allowdisplaybreaks

\let\pa\partial

\newcommand{\N}{{\mathbb N}}
\newcommand{\R}{{\mathbb R}}

\newtheorem{theorem}{Theorem}
\newtheorem{lemma}[theorem]{Lemma}
\newtheorem{proposition}[theorem]{Proposition}
\newtheorem{remark}[theorem]{Remark}

\let\ga=\gamma
\let\de=\delta

\let\la=\lambda

\let\pa=\partial

\let\om=\omega

\begin{document}

\title[Landau equation]{Solving Linearized Landau Equation Pointwisely}

\author[HT. Wang]{Haitao Wang}
\address{Haitao Wang, Institute of Natural Sciences and School of Mathematical Sciences, Shanghai Jiao Tong University, Shanghai, China}
\email{haitaowang.math@gmail.com}

\author[K.-C. Wu]{Kung-Chien Wu}
\address{Kung-Chien Wu, Department of Mathematics,
National Cheng Kung University, Tainan, Taiwan AND National Center for Theoretical Sciences, National Taiwan University, Taipei, Taiwan}
\email{kungchienwu@gmail.com}

\date{\today}

\thanks{The authors would like to thank Professor Tai-Ping Liu
for his encouragement and fruitful discussions concerning this project. Part of this work was done while the first author was in Institute of Mathematics, Academia Sinica, Taiwan. The second author is supported by the Ministry of Science and Technology under the grant 104-2628-M-006-003-MY4 and National Center for Theoretical Sciences.}

\begin{abstract}
We study the pointwise (in the space and time variables) behavior of the linearized Landau equation for hard and moderately soft potentials. The solution has very clear description in the $(x,t)-$variables, including the large time behavior and the asymptotic behavior. More precisely, we obtain the pointwise fluid structure inside the finite Mach number region, and exponential or sub-exponential decay, depending on interactions between particles, in the space variable outside  the finite Mach number region. The spectrum analysis, regularization effect and refined weighted energy estimate play important roles in this paper.
\end{abstract}



\maketitle
\tableofcontents


\section{Introduction}
\subsection{The models}
The generalized Landau equation reads
\begin{equation}\label{in.1.a}
\left\{\begin{array}{l}
\displaystyle \pa_{t}F+\xi\cdot\nabla_{x}F= Q(F,F)\,,
\\ \\
\displaystyle F(0,x,\xi)=F_{0}(x,\xi)\,,
\end{array}
\right.\end{equation}
where $F(t, x, \xi)$ is the distribution function for the particles at time $t$, position $x=(x_{1},x_{2},x_{3})\in\R^{3}$ and velocity $\xi=(\xi_{1},\xi_{2},\xi_{3})\in\R^{3}$. The operator $Q(\cdot,\cdot)$ is the so-called Landau collision operator given by
$$
Q(f,g)=\frac{1}{2}\nabla_{\xi}\cdot\Big\{\int_{\R^{3}}\Phi(\xi-\xi_{*})
\big[f(\xi_{*})\nabla_{\xi}g(\xi)+g(\xi_{*})\nabla_{\xi}f(\xi)-f(\xi)\nabla_{\xi_{*}}g(\xi_{*})-g(\xi)\nabla_{\xi_{*}}f(\xi_{*})\big]d\xi_{*}\Big\}\,.
$$
The positive semi-definite matrix $\Phi(\xi)$ has the general form
$$
\Phi(\xi)=B(|\xi|)S(\xi)\,,
$$
where $B(|\xi|)=|\xi|^{\ga+2}$ is a function depending on the nature of the interaction between the particles, and $S(\xi)$ is the 3 by 3 matrix
$$
S(\xi)=I_{3}-\frac{\xi\otimes \xi}{|\xi|^{2}}\,.
$$
This leads to the usual classification in terms of hard potential $(0<\ga\leq 1)$, Maxwellian molecules $(\ga=0)$, moderately soft potential ($-2\leq \gamma<0$) or
very soft potential $(-3<\ga<-2)$. In particular, $\ga=-3$ corresponds to the important
Coulomb interaction in plasma physics.
Just as for the Boltzmann equation, little is known for soft potentials, i.e., $\ga<0$, and even less for very
soft potentials, i.e., $\ga<-2$. In this paper, we will consider the case $-2\leq \ga\leq 1$.

Similar to the Boltzmann equation, the Maxwellians are steady states to the Landau equation (\ref{in.1.a}). Thus, it is natural to investigate the behavior of the solution near Maxwellian. This leads us to  linearize the Landau equation \eqref{in.1.a} around a normalized global Maxwellian $\mu(\xi)$,
$$
\mu(\xi)=\frac{1}{(2\pi)^{3/2}}\exp\Big(\frac{-|\xi|^{2}}{2}\Big)\,,
$$
with the standard perturbation $f(t,x,\xi)$ to $\mu$ as
$$
F=\mu+\mu^{1/2}f\,.
$$
From the fact that $Q(\mu,\mu)=0$, we have
$$
Q(\mu+\mu^{1/2}f, \mu+\mu^{1/2}f)=2Q(\mu,\mu^{1/2}f)+Q(\mu^{1/2}f,\mu^{1/2}f)\,.
$$
By dropping the nonlinear term, we can define the linearized Landau collision operator $L$ as
\begin{equation}
Lf=2\mu^{-1/2}Q(\mu,\mu^{1/2}f)\,.
\end{equation}
The linearized Landau equation for $f(t,x,\xi)$ now takes the form
\begin{equation}\label{in.1.c}
\left\{\begin{array}{l}
\displaystyle\pa_{t}f+\xi\cdot\nabla_{x}f=Lf\,,
\\ \\
f(0,x,\xi)=f_{0}(x,\xi)\,.
\end{array}
\right.\end{equation}
Here and below we define $f(t,x,\xi)=\mathbb{G}^{t}f_{0}(x,\xi)$, i.e., $\mathbb{G}^{t}$ is the solution operator (Green's function) of the linearized Landau equation (\ref{in.1.c}). In this paper, we will study the pointwise structure of the linearized Landau equation (\ref{in.1.c}).

It is well-known that the null space of $L$ is a five-dimensional vector space with the orthonormal basis $\{\chi_{i}\}_{i=0}^{4}$,
where
$$
Ker(L)=\left\{\chi_{0}, \chi_{i}, \chi_{4}\right\}=\left\{\mu^{1/2},\xi_{i}\mu^{1/2}, \frac{1}{\sqrt{6}}(|\xi|^{2}-3)\mu^{1/2}\right\}\,, \quad i=1,2,3\,.
$$
Based on this property, we can introduce the Macro-Micro decomposition: let $\mathrm{P}_0$ be the orthogonal projection with respect to the $L_{\xi}^{2}$ inner product onto ${\rm Ker}(L)$, and $\mathrm{P}_1\equiv \mathrm{Id}-\mathrm{P}_0$.

\subsection{Main results}

Before the presentation of the main theorem, let us define some notations in
this paper. We denote $\left\langle \xi\right\rangle ^{s}=(1+|\xi|^{2})^{s/2}$,
$s\in{\mathbb{R}}$. For the microscopic variable $\xi$, we denote
\[
|g|_{L_{\xi}^{2}}=\Big(\int_{{\mathbb{R}}^{3}}|g|^{2}d\xi\Big)^{1/2}\,,
\]
and the weighted norms $|g|_{L_{\xi}^{2}(m)}$ by
\[
|g|_{L_{\xi}^{2}(m)}=\Big(\int_{{\mathbb{R}}^{3}}|g|^{2}md\xi\Big)^{1/2}\,.
\]
The $L_{\xi}^{2}$ inner product in ${\mathbb{R}}^{3}$ will be denoted by
$\big<\cdot,\cdot\big>_{\xi}$,
\[
\left<f,g\right>_\xi=\int f(\xi)\overline{g(\xi)}d\xi.
\]
For the space variable $x$, we have similar
notations. In fact, $L_{x}^{2}$ is the classical Hilbert space with norm
\[
|g|_{L_{x}^{2}}=\Big(\int_{{\mathbb{R}^{3}}}|g|^{2}dx\Big)^{1/2}\,.
\]
We denote the sup norm as
\[
|g|_{L_{x}^{\infty}}=\sup_{x\in{\mathbb{R}^{3}}}|g(x)|\,.
\]
The standard vector product will be denoted by $(a,b)$ or $a\cdot b$, for any vectors $a,b\in\R^{3}$. For any vector
function $u\in L_{\xi}^{2}$, $\mathbb{P}(\xi)u$ denotes the orthogonal projection
along the direction of vector $\xi$, i.e.,
\[
\mathbb{P}(\xi)u=\frac{u\cdot \xi}{|\xi|^{2}}\xi\,.
\]
For the Landau equation, the natural norm in $\xi$ is $|\cdot|_{L_{\sigma}^{2}}%
$, which is defined by
\[
|g|_{L_{\sigma}^{2}}^{2}=|\left\langle \xi\right\rangle ^{\frac{\ga+2}{2}%
}g|_{L_{\xi}^{2}}^{2}+|\left\langle \xi\right\rangle ^{\frac{\ga}{2}}%
\mathbb{P}(\xi)\nabla_{\xi}g|_{L_{\xi}^{2}}^{2}+\big|\left\langle \xi\right\rangle
^{\frac{\ga+2}{2}}\big[I_{3}-\mathbb{P}(\xi)\big]\nabla_{\xi}g\big|_{L_{\xi}^{2}%
}^{2}
\]
and
\[
|g|_{L_{\sigma}^{2}(m)}^{2}=|\left\langle \xi\right\rangle ^{\frac{\ga+2}{2}%
}g|_{L_{\xi}^{2}(m)}^{2}+|\left\langle \xi\right\rangle ^{\frac{\ga}{2}}%
\mathbb{P}(\xi)\nabla_{\xi}g|_{L_{\xi}^{2}(m)}^{2}+\big|\left\langle \xi\right\rangle
^{\frac{\ga+2}{2}}\big[I_{3}-\mathbb{P}(\xi)\big]\nabla_{\xi}g\big|_{L_{\xi}^{2}(m)%
}^{2}
\]
Moreover, we define
\[
\Vert g\Vert_{L^{2}}^{2}=\int_{{\mathbb{R}^{3}}}|g|_{L_{\xi}^{2}}^{2}dx\,,\quad\Vert
g\Vert_{L^{2}(m)}^{2}=\int_{{\mathbb{R}^{3}}}|g|_{L_{\xi}^{2}(m)}^{2}dx\,,
\]%
\[
\Vert g\Vert_{L_{\sigma}^{2}}^{2}=\int_{{\mathbb{R}^{3}}}|g|_{L_{\sigma}^{2}}%
^{2}dx\,, \quad\Vert g\Vert_{L_{\sigma}^{2}(m)}^{2}=\int_{{\mathbb{R}^{3}}}|g|_{L_{\sigma}^{2}(m)}%
^{2}dx\,.
\]
and
\[
\Vert g\Vert_{L_{x}^{\infty}L_{\xi}^{2}}=\sup_{x\in{\mathbb{R}^{3}}}|g|_{L_{\xi}^{2}%
}\,,\quad\Vert g\Vert_{L_{x}^{1}L_{\xi}^{2}}=\int_{{\mathbb{R}^{3}}}|g|_{L_{\xi}^{2}%
}dx\,.
\]
Finally, we define the high order Sobolev norm: let $k\in{\mathbb{N}}$ and $\alpha$ be any multi-index with $|\alpha|\leq k$,
\[
\left\Vert g\right\Vert _{H^{k}_{x}L^{2}_{\xi}(m)}=\sum_{|\alpha|\leq k}\left\Vert
\pa_{x}^{\alpha}g\right\Vert _{L^{2}(m)}\,.
\]

The domain decomposition plays an important role in our analysis, hence we
need to define a cut-off function $\chi:{\mathbb{R}}\rightarrow{\mathbb{R}}$,
which is a smooth non-increasing function, $\chi(s)=1$ for $s\leq1$,
$\chi(s)=0$ for $s\geq2$ and $0\leq\chi\leq1$. Moreover, we define $\chi
_{R}(s)=\chi(s/R)$.

For simplicity of notations, hereafter, we abbreviate ``{ $\leq C$} " to ``{ $ \lesssim$ }", where $C$ is a positive constant depending only on fixed numbers.

The precise description of our main result is as follows:
\begin{theorem}
\label{thm:main} Let $f$ be a solution to \eqref{in.1.c} with initial data
compactly supported in $x$-variable and bounded in the weighted $\xi$-space
\[
f_{0}(x,\xi)\equiv0,\;\text{ for }\left\vert x\right\vert \geq1.
\]
There exists a positive constant $M$ such that the following hold for $t\geq1$:

\begin{enumerate}
\item For $-1\leq\gamma\leq1$, there exists a positive constant $C$ such that
the solution to \eqref{in.1.c} satisfies

\begin{enumerate}
\item For $\left< x\right> \leq2Mt$,
\[
\left| f(t,x)\right| _{L_{\xi}^{2}} \leq C \left[
\begin{array}
[c]{l}%
\left( 1+t\right) ^{-2}e^{-\frac{\left( \left| x\right| -\mathbf{c}t\right)
^{2}}{Ct}}+\left( 1+t\right) ^{-3/2}e^{-\frac{\left| x\right| ^{2}}{Ct}%
}\\[2mm]%
+\mathbf{1}_{\{\left| x\right| \leq\mathbf{c}t\}} \left( 1+t\right)
^{-3/2}\left( 1+\frac{\left| x\right| ^{2}}{1+t}\right) ^{-3/2}+e^{-t/C}%
\end{array}
\right]  |||f_{0}|||\,.
\]

\item For $\left< x\right> \geq2Mt$,
\[
\left| f (t,x) \right| _{L^{2}_{\xi}} \leq C e^{-\left( \left< x\right>
+t\right) /C} |||f_{0}|||\,.
\]

\end{enumerate}

\item For $-2\leq\gamma<-1$, for any given positive integer $N$, and any given
sufficiently small $\alpha>0$, there exist positive constants $C$, $C_{N} $
and $c_{\alpha}$ such that the solution to \eqref{in.1.c} satisfies

\begin{enumerate}
\item For $\left< x\right> \leq2Mt$,
\[
\left| f(t,x)\right| _{L_{\xi}^{2}}\leq C_{N} \left[
\begin{array}
[c]{l}%
\left( 1+t\right) ^{-2}\left( 1+\frac{\left( \left| x\right| -\mathbf{c}%
t\right) ^{2}}{1+t}\right) ^{-N}+\left( 1+t\right) ^{-3/2}\left(
1+\frac{\left| x\right| ^{2}}{1+t}\right) ^{-N}\\[2mm]%
+\mathbf{1}_{\{\left| x\right| \leq\mathbf{c}t\}}\left( 1+t\right) ^{-3/2}
\left( 1+\frac{\left| x\right| ^{2}}{1+t}\right) ^{-3/2}+ e^{-t/C_{N}}%
\end{array}
\right]  |||f_{0}|||\,.
\]

\item For $\left< x\right> \geq2Mt$,
\[
\left|  f(t,x) \right| _{L^{2}_{\xi}}\leq C(1+t) e^{-c_{\alpha}(\left<
x\right> +t)^{\frac{2}{1-\gamma}}} \left\| f_{0} \right\| _{L^{\infty}_{x}
L^{2}_{\xi}\left( e^{\alpha|\xi|^{2}}\right) } .
\]

\end{enumerate}
\end{enumerate}

Here $\mathbf{1}_{\{\cdot\}}$ is the indicator function,
\begin{align*}
|||f_{0}|||\equiv\max\left\{  \left\| f_{0}\right\| _{L^{2} (\mathcal{M}%
)},\left\| f_{0}\right\| _{L^{1}_{x} L^{2}_{\xi}} \right\} ,
\end{align*}
and
\begin{align*}
\mathcal{M}\equiv%
\begin{cases}
1 & \ga\in\left[ 0,1\right] ,\\
\left< \xi\right> ^{2 |\ga|} & \ga \in\left[ -2,0\right) .
\end{cases}
\end{align*}
The constant $\mathbf{c}=\sqrt{5/3}$ is the sound speed associated with the
normalized global Maxwellian.
\end{theorem}

\subsection{Review of previous works and significant points of the paper}

Let us give an overview of the previous works on the Cauchy theory for the Landau equation in a close-to-equilibrium framework. We refer to Alexandre and Villani \cite{[2]}
for the existence of renormalized solutions, to Desvillette and Villani \cite{[18]}
for conditionally almost exponential convergence towards equilibrium and to a
recent work by Carrapatoso, Tristani and Wu \cite{[CTK]} for exponential decay
towards equilibrium when initial data are close enough to equilibrium.
Moreover, Guo \cite{[Guo]} and Strain and Guo \cite{[31], [32]} developed an
existence and convergence towards equilibrium theory based on energy method
for initial data close to the equilibrium state in some Sobolev spaces. Recently
the set of initial data for which this theory is valid has been enlarged by
Carrapatoso and Mischler \cite{[9]} via a linearization method.

In this paper, we study the linearized Landau equation with hard or moderately soft potentials ($\gamma\in [-2,1]$) in
the close to equilibrium setting. In the literature, this kind of problem
basically focuses on the rate of convergence to equilibrium (see the reference
listed above). In contrast, in this paper we supply a very explicit
description of the solution in the sense of pointwise estimate.

Here are some significant points of the paper:
\begin{itemize}
\item We give a complete pointwise description of the solution, which consists of two parts: inside the finite Mach number region (the large time behavior) and
outside the finite Mach number region (the asymptotic behavior).

\begin{enumerate}

\item Concerning the solution inside the finite Mach number region (the large time behavior), thanks to the spectrum analysis \cite{YangYu} and our
generalization (Lemma \ref{lem:smooth}), we have pointwise fluid structure, which
is much richer than previous results. The leading terms of wave
propagation have been recognized. More precisely, they are characterized by the Huygens waves, the diffusion waves and the Riesz waves.

\item Concerning the solution outside the finite Mach number region (the asymptotic behavior), we have exponential decay
for $-1\leq\gamma\leq 1$ and sub-exponential decay for $-2\leq
\gamma<-1$ in the space variable $x$. The results
are consistent with the wave behaviors inside the finite Mach number region respectively.
We believe this is the first result for the asymptotic behavior of the
Landau kinetic equation.
\end{enumerate}

\medskip

\item  The pointwise behavior of leading fluid part is determined by how the solution depends on the Fourier transformed variable of the spatial variable, i.e., smoothly or analytically (see Proposition \ref{flu_smooth}). This connection has been investigated for various models, for example, the compressible Navier-Stokes equation\cite{Zeng1994,[LiuWang]}, the one space dimensional Boltzmann equation with hard sphere \cite{[LiuYu]} and with hard potential \cite{[LeeLiuYu]} and the three space dimensional Boltzmann equation with hard sphere \cite{[LiuYu1],[LiuYu2]}. It was noted in \cite{[LiuYu2]} that the wave patterns in 3D case are much richer than 1D case. In addition to the Huygens waves in 3D case, which is the counterpart of the diffusion waves in 1D case, they also contain the diffusion waves and the Riesz waves (those decay slowly inside the finite Mach number region). The identification of those wave patterns heavily relies on the decomposition and recombination of eigenvalues near the origin and associated eigenfuntions. The whole procedure makes use of their analytic dependencies on the Fourier variables (see sections 7.4--7.5 \cite{[LiuYu2]} for details).  However, for the Landau equation with $\gamma\in[-2,-1)$, the dependence is no more analytic. We fully exploit the symmetry properties of eigenvalues and eigenfunctions and thus establish the appropriate decomposition (the result is summarized in Lemma \ref{lem:smooth}), so that we can obtain the pointwise estimate for this case. The method used here can be generalized to deal with other kinetic models.

\medskip

\item The regularization estimate plays a crucial role in this paper (see
Lemma \ref{improve}), this enables us to
obtain the pointwise estimate without regularity assumption on the initial
data.
In the literature, the regularization estimates for the kinetic Fokker-Planck
equation and the Landau equation have been proved for various purposes, see for
instance \cite{[Herau]}, \cite{[MisMou]}, \cite{Villani} (Appendix A.21.2) for the
Fokker-Planck case and \cite{[CTK]} for the Landau case. In this paper, we
construct the regularization estimates in suitable weighted function space (more
precisely, suitable for outside the finite Mach number region), the calculation of the estimate is
interesting and  more sophisticated than before. Moreover, this type of
regularization estimate is itself new for the Landau equation. \medskip

\item The pointwise estimate of the solution outside the finite Mach number region is constructed by the weighted energy estimate. The time dependent weight functions are chosen according to different interactions between particles. For $-1\leq \gamma \leq 1$, the weight function depends only on the time and the space variables, and exponentially grows in space. Since it commutes with the operator $K$, the estimate is relatively simple. However, for $-2\leq \gamma <-1$, the weight function is much more complicated. Indeed, it depends on the velocity variable as well and thus does not commute with the operator $K$, which leads to the coercivity of linearized collision operator cannot be applied directly and loss of control of some terms at first glance. The difficulty is eventually overcome by fine tuning the weight functions, introducing refined space-velocity domain decomposition and analyzing the operator $K$ with weight accordingly (see Proposition \ref{weig_2}).
\end{itemize}

For the pointwise behavior of
the kinetic type equation, the Boltzmann equation for hard sphere and hard
potential with cutoff (see \cite{[LeeLiuYu],[LiuYu],[LiuYu2], [LiuYu1]}) should be mentioned; those works are some of the most important results in kinetic theory. Let us point out the similarities and
differences between the Landau kinetic equation and the Boltzmann
equation for hard sphere or hard potential with cutoff.

\begin{itemize}
\item The solutions of both in large time are dominated by the fluid parts. To extract them, both need the long wave-short wave decomposition. In 3D, for Landau with $\gamma\in[-1,1]$ and Boltzmann with hard sphere, the fluid parts are similar, they are characterized by  Huygens waves, diffusion waves and Riesz waves. Moreover, the former two waves are of exponential type while the Riesz waves are of the algebraic type. By comparison, the exponential type are replaced by algebraic type for Landau with $\gamma\in[-2,-1)$, and the Riesz waves remain the same. The fluid behavior can be seen formally from the
Chapman-Enskog expansion, which indicates that the macroscopic part (the fluid part) of solution satisfies the viscous conservation laws system. For both Boltzmann and Landau there are conservation laws of mass, momentum and energy; this explains the wave structures of fluid parts. This picture is valid even for some more general kinetic equation as well since physically the kinetic model can be approximated by fluid
equations in the long term. \medskip

\item Since the leading terms of solution in large time are fluid parts and
they essentially have finite propagation speed, the solution outside the finite Mach number region
 should be insignificant. In fact, it is shown that the asymptotic
behaviors exponentially or sub-exponentially decay. This is similar to the
solution of the Boltzmann equation outside the finite Mach number region. \medskip

\item The regularization mechanism of the Landau equation is distinct from that of the Boltzmann equation. For the Boltzmann equation, the initial singularity will be preserved
(although it decays in time very fast). Thus, one has to single out the
singular kinetic wave, and then the space regularity of resultant remainder
part comes from the transport term and the compact part of the collision
operator (see the "Mixture Lemma" in \cite{[LeeLiuYu]}, \cite{[LiuYu]} and
\cite{[LiuYu2]}). For the Landau equation, however, the solution becomes
smooth in the space variable immediately. This property comes from the
combined effect of the ellipticity in the velocity variable $\xi
$ and the transport term (see Lemma \ref{improve}).
\end{itemize}

\subsection{Method of proof and plan of the paper}
The main idea of this paper is to combine the long wave-short wave
decomposition, the weighted energy estimate
and the regularization estimate together to analyze the solution. The long-short wave decomposition, which is based on the Fourier transform, gives the fluid structure of the solution. The weighted energy estimate is for the pointwise estimate of the solution outside the finite Mach number region, in which regularization is used. We explain the idea as below.

For the region inside the finite Mach number region, the solution is
dominated by the fluid part, which is contained in the long wave part. In order to obtain its estimate, we devise different methods for $-1\leq\gamma\leq 1$ and $-2\leq\gamma<-1$ respectively. Taking advantage of spectrum
information of the Landau collision operator \cite{YangYu} (in fact, we need more information about analyticity or smoothness of the operator, which can be found in Lemma \ref{lem:smooth}), the complex analytic ($-1\leq\gamma\leq 1$) or the Fourier multiplier ($-2\leq\gamma<-1$) techniques can
be applied to obtain pointwise structure of the fluid part. The regularization estimate together with $L^{2}$ decay of the
short wave yields the $L^{\infty}$ decay of the short wave, this finishes the pointwise structure inside the finite Mach number region.

Note that to complete the structure outside the finite Mach number region, the weighted energy estimates (Proposition \ref{weig_1} and Proposition \ref{weig_2}) come to
play a role. The weighted functions are carefully chosen for different
$\gamma$'s. It is noted that the sufficient understanding of the structure inside the finite Mach number region, which has been obtained previously, is needed in the estimate.
And the regularization effect (Lemma
\ref{improve}) makes it possible to do the higher order
weighted energy estimates without any regularity assumption of the initial condition. Then the desired pointwise estimate follows from
Sobolev inequality.

It is worth making a comment for the case $-3\leq\gamma<-2$. Due to the weak coercivity, the linearized collision operator does not satisfy any spectral gap inequality and the detailed spectral information such as Lemma \ref{pr12} and Lemma \ref{lem:smooth} are absent. As a consequence, the pointwise description of fluid structure is too much to hope for. However, it is possible to obtain the time decay of the solution inside the finite Mach number region by Kawashima's
moments method \cite{[Kawashima]} and Strain's interpolation argument \cite{[Strain]}. On the other hand, we realize that the regularization estimate and the weighted energy estimates work as well; this recognize the behavior outside the finite Mach number region. Nevertheless, since our main concerns are the pointwise estimate and the explicit wave propagation, we choose to omit this part in our paper.

The rest of this paper is organized as follows: We first prepare some important properties in section \ref{pre} for the spectrum analysis and regularization estimates.
Then we study the solution inside the finite Mach number region in section \ref{inside} and outside the finite Mach number region in section \ref{outside}. Finally, we
prove Lemma \ref{lem:smooth} and Lemma \ref{improve} in sections \ref{Pf-smooth} and \ref{Hypo} respectively, those are key lemmas in this paper.

\section{Preliminaries}\label{pre}
In this section, we will prepare some important properties, including linearized collision operator, spectrum analysis of the collision operator and the regularization estimate.

We can decompose the linearized collision operator $L$ as the diffusion operator and the integral operator,
$$
Lf= \widetilde{\Lambda}f +\widetilde{K}f\,.
$$
The diffusion operator $ \widetilde{\Lambda}$ is
$$
\widetilde{\Lambda}f =\mu^{-1/2}\nabla_{\xi}\cdot\big[(\sigma \mu)\nabla_{\xi}(\mu^{-1/2}f)\big]\,,
$$
where the symmetric matrix $\sigma(\xi)$ is defined by
\begin{equation}\begin{array}{l}
\displaystyle \sigma(\xi)=\int_{\R^{3}}\Phi(\xi-\xi_{*})\mu(\xi_{*})d\xi_{*}\,,
\end{array}\end{equation}
and the integral operator $\widetilde{K}$ is
$$
\widetilde{K}f=\int_{\R^{3}}\mu^{-1/2}(\xi)\mu^{-1/2}(\xi_{*})(\nabla_{\xi},\nabla_{\xi_{*}}\cdot Z(\xi,\xi_{*}))f(\xi_{*})d\xi_{*}\equiv
\int_{\R^{3}}\tilde{k}(\xi,\xi_{*})f(\xi_{*})d\xi_{*}\,,
$$
where
$$
Z(\xi,\xi_{*})=\mu(\xi)\mu(\xi_{*})\Phi(\xi-\xi_{*})\,.
$$

Let us list some fundamental properties of the linearized Landau collision operator $L$:
\begin{lemma} \label{prop1}{\rm\cite{[Degond], [Guo]}} For any $\ga\geq -3$,
\begin{enumerate}[(i).]
\item
$$
\widetilde{\Lambda} g=\nabla_{\xi}\cdot\big[\sigma\nabla_{\xi}g\big]+ \frac{1}{2}\left(\nabla_{\xi}\cdot\big[\sigma\xi\big]-\frac{1}{2}(\xi,\sigma\xi )\right)g\,.
$$
\item The spectrum of $\sigma(\xi)$ consists of a simple eigenvalue $\la_{1}(\xi)>0$ associated with the eigenvector $\xi$,
and a double eigenvalue $\la_{2}(\xi)>0$ associated with eigenvectors $\xi^{\perp}$. Immediately, for any vector function $u$, we have
\begin{align*}
\nabla_{\xi}\cdot\sigma&=-\sigma\xi=- \la_{1}\xi\,,\quad
 (\xi,\sigma\xi )=\la_{1}(\xi)|\xi|^{2}\,, \\
(u,\sigma u )&=\la_{1}(\xi)|\mathbb{P}(\xi)u|^{2}+\la_{2}(\xi)\big|\big[I_{3}-\mathbb{P}(\xi)\big]u\big|^{2}\\
&\geq c_0 \left\{\left<\xi\right>^{\ga} \left|\mathbb{P}(\xi)u\right|^2 + \left<\xi\right>^{\ga+2} \left| \left[I_{3}-\mathbb{P}(\xi)\right]u\right|^2	 \right\}\,.
\end{align*}
Moreover, as $|\xi|\to \infty$, we have
$$
\la_{1}(\xi)\sim 2\left< \xi \right>^{\ga},\quad \la_{2}(\xi)\sim\left< \xi \right>^{\ga+2}\,.
$$

\item For any multi-index $k$, we have
$$
|\pa^{k}_{\xi}\sigma|\lesssim\left< \xi \right>^{\ga+2-|k|}\,,\quad |\pa_{\xi}^{k}(\sigma\xi)|\lesssim\left< \xi \right>^{\ga+1-|k|}\,,
$$
and for $|\xi|\to \infty$,
$$
|\pa_{\xi}^{k}\la_{1}(\xi)|\lesssim\left< \xi \right>^{\ga-|k|}\,,\quad
|\pa_{\xi}^{k}\la_{2}(\xi)|\lesssim\left< \xi \right>^{\ga+2-|k|}\,.
$$
\item Let
$$
\Lambda g=\widetilde{\Lambda}g-\varpi \chi_{R} g\,, \quad Kg=\widetilde{K}g+\varpi \chi_{R} g\,,
$$
where $\varpi, R$ are large enough, then
\begin{align*}
\big<-\Lambda g, g\big>_{\xi}\geq c_{0}|g|^{2}_{L^{2}_{\sigma}}
\end{align*}
and
\begin{align*}
\big<Kg, g\big>_{\xi}\leq |g|^{2}_{L^{2}_{\xi}}
\,.
\end{align*}
\item (Coercivity) There exists $\nu_{0}>0$ such that%
\begin{equation}
\left\langle -Lg,g\right\rangle _{\xi}\geq \nu_{0}\left\vert \mathrm{P}%
_{1}g\right\vert _{L_{\sigma}^{2}}^{2}. \label{coercivity}%
\end{equation}

\end{enumerate}
\end{lemma}

The following lemma, which will be used in Proposition \ref{weig_1}, is a
consequence of (\ref{coercivity}).

\begin{lemma}\label{lem-a}
Let $\ga\geq -2$, there exist $C_{1}, C_{2}$ such that
$$
\big<|\xi|^{\ga+2}g,g\big>_{\xi}\leq C_{1} \big<-Lg,g\big>_{\xi}+C_{2}|g|^{2}_{L^{2}_{\xi}}\,.
$$
\end{lemma}

In order to study the solution inside the finite Mach number region, we need to recall the
spectrum $\rm{Spec}(\eta)$, $\eta \in\mathbb{R}^3$, of the operator $-i  \xi\cdot\eta +L$.
\begin{lemma}\label{pr12}{\rm\cite{[Degond], YangYu}}
Set $\eta=|\eta|\om$. For any $\ga\geq-2$, there exist $\de_{0}>0$ and $\tau=\tau(\de_{0})>0$ such that
\begin{enumerate}[(i).]
\item  For any $|\eta|>\de_{0}$,
\begin{equation*}
\hbox{\rm{Spec}}(\eta )\subset\{z\in\mathbb{C} : Re(z)<-\tau\}\,.
\end{equation*}
\item For any $|\eta |<\de_{0}$, the spectrum within the region $\{z\in\mathbb{C} : Re(z)>-\tau\}$ consists of exactly five eigenvalues $\{\varrho_{j}(\eta  )\}_{j=0}^{4}$,
\begin{equation*}
\hbox{\rm{Spec}}(\eta )\cap\{z\in\mathbb{C} : Re(z)>-\tau\}=
\{\varrho_{j}(\eta  )\}_{j=0}^{4}\,,
\end{equation*}
and corresponding eigenvectors $\{e_{j}(\eta  )\}_{j=0}^{4}$, where
\begin{equation*}\begin{array}{l}
\displaystyle \varrho_{j}( \eta   )=-i\, a_{j}|\eta |-A_{j}|\eta |^{2}+O(| \eta   |^{3})\,,
\\ \\
\displaystyle e_{j}( \eta   )=E_{j}+O(| \eta   |)\,,
\end{array}
\end{equation*}
here $A_{j}>0$, $\big<e_{j}(- \eta  ),e_{l}( \eta  )\big>_{\xi}=\de_{jl}$, $1\leq j,l\leq 3$ and
\begin{equation*}\label{speed}
\left\{\begin{array}{l}
a_{0}=\sqrt{\frac{5}{3}}\,,\quad a_{1}=-\sqrt{\frac{5}{3}}\,,\quad a_{2}=a_{3}=a_{4}=0 \,,
\\[2mm]
E_{0}=\sqrt{\frac{3}{10}}\chi_{0}+\sqrt{\frac{1}{2}}\om\cdot \overline{\chi}+\sqrt{\frac{1}{5}}\chi_{4}\,,
\\[2mm]
E_{1}=\sqrt{\frac{3}{10}}\chi_{0}-\sqrt{\frac{1}{2}}\om\cdot \overline{\chi}+\sqrt{\frac{1}{5}}\chi_{4}\,,
\\[2mm]
E_{2}= -\sqrt{\frac{2}{5}}\chi_{0}+\sqrt{\frac{3}{5}}\chi_{4}\,,
\\[2mm]
E_{3}= \om_{1}\cdot \overline{\chi}\,,
\\[2mm]
E_{4}= \om_{2}\cdot \overline{\chi}\,.
\end{array}
\right.\end{equation*}
where $\overline{\chi}=(\chi_{1},\chi_{2},\chi_{3})$, and $\{\om_{1},\om_{2}, \om\}$ is an orthonormal basis of $\R^{3}$. More precisely, the semigroup $e^{(-i \xi\cdot\eta +L )t}$ can be decomposed as
\begin{align*}
\displaystyle e^{(-i \xi\cdot\eta +L )t}g
&=e^{(-i \xi\cdot\eta +L)t}\Pi_{\eta}^{\perp}g\nonumber
\\
&\quad+\textbf{1}_{\{| \eta  |<\de_{0}\}}\sum_{j=0}^{4}e^{\varrho_{j}( \eta  )t}\big<e_{j}(- \eta   ), g\big>_{\xi}e_{j}( \eta   )\,.
\end{align*}
where $\textbf{1}_{\{\cdot\}}$ is the indicator function, and there exists $C>0$ such that
$$
\left|e^{(-i \xi\cdot\eta +L)t}\Pi_{\eta}^{\perp}g \right|_{L^{2}_{\xi}}\leq e^{-Ct}|g|_{L^{2}_{\xi}}\,.
$$
\end{enumerate}
\end{lemma}
Note that the parameters $A_{j}, 0\leq j\leq 4$ with the relations $A_{0}=A_{1}$ and $A_{3}=A_{4}$, are the dissipation parameters corresponding to the Chapman-Enskog expansion relating the Landau equation to the Navier-Stokes equation.

Furthermore, we have more detailed information about the smooth and analytic properties of  eigen-pairs $\{(\varrho_j,e_j)\}_{j=0}^{4}$, which is essentially used in obtaining the estimate of fluid structure (Proposition \ref{flu_smooth}).
\begin{lemma}\label{lem:smooth}
	For $\eta\in\mathbb{R}^{3}$ with $\left|\eta\right|\ll1$,	
	\begin{align*}
		\varrho_{0}\left(\eta\right) & =-i\left|\eta\right|\left(c+\mathscr{A}_{0}\left(\left|\eta\right|^{2}\right)\right)-A_{0}\left|\eta\right|^{2}+\mathscr{A}_{1}\left(\left|\eta\right|^{2}\right),\\
		\varrho_{1}\left(\eta\right) & =i\left|\eta\right|\left(c+\mathscr{A}_{0}\left(\left|\eta\right|^{2}\right)\right)-A_{0}\left|\eta\right|^{2}+\mathscr{A}_{1}\left(\left|\eta\right|^{2}\right),\\
		\varrho_{2}\left(\eta\right) & =-A_{2}\left|\eta\right|^{2}+\mathscr{A}_{2}\left(\left|\eta\right|^{2}\right),\\
		\varrho_{3}\left(\eta\right)=\varrho_{4}\left(\eta\right) & =-A_{3}\left|\eta\right|^{2}+\mathscr{A}_{3}\left(\left|\eta\right|^{2}\right),
	\end{align*}
	for some analytic (smooth) functions $\mathscr{A}_{j}:\mathbb{R}\to\mathbb{R}$,
	$j=0,1,2,3$ when $-1\leq\gamma\leq1$ ($-2\leq\gamma<-1$) respectively.  $\varrho_j$, $j=2,3,4$ are actually real-valued functions.  For $\left|s\right|\ll1$,
	\[
	\mathscr{A}_{0}\left(s\right)=O\left(s\right),\quad\mathscr{A}_{j}\left(s\right)=O\left(s^{2}\right),\quad j=1,2,3.
	\]
	Furthermore, there exist analytic (smooth) functions $\mathtt{a}_{k.l}^{j}:\mathbb{R}\to\mathbb{R}$
	such that when $-1\leq\gamma\leq1$ ($-2\leq\gamma<-1$),
	
	\[
	\begin{aligned}e_{0}\left(\eta\right)= & \mathscr{L}_{0}\left[\left(\mathtt{a}_{0,1}^{0}(\left|\eta\right|^{2})+i\left|\eta\right|\textrm{a}_{0,2}^{0}(\left|\eta\right|^{2})\right)\chi_{0}\right.\\
	& \left.\qquad+\left(\mathtt{a}_{0,1}^{1}(\left|\eta\right|^{2})+i\left|\eta\right|\mathtt{a}_{0,2}^{1}(\left|\eta\right|^{2})\right)
\om\cdot \overline{\chi}\right.\\
	& \left.\qquad+\left(\mathtt{a}_{0,1}^{4}(\left|\eta\right|^{2})+i\left|\eta\right|\mathtt{a}_{0,2}^{1}(\left|\eta\right|^{2})\right)\chi_{4}\right],
	\end{aligned}
	\]
	\[
	\begin{aligned}e_{1}\left(\eta\right)= & \mathscr{L}_{1}\left[\left(\mathtt{a}_{0,1}^{0}(\left|\eta\right|^{2})-i\left|\eta\right|\textrm{a}_{0,2}^{0}(\left|\eta\right|^{2})\right)\chi_{0}\right.\\
	& \left.\qquad-\left(\mathtt{a}_{0,1}^{1}(\left|\eta\right|^{2})-i\left|\eta\right|\mathtt{a}_{0,2}^{1}(\left|\eta\right|^{2})\right)\om\cdot \overline{\chi}\right.\\
	& \left.\qquad+\left(\mathtt{a}_{0,1}^{4}(\left|\eta\right|^{2})-i\left|\eta\right|\mathtt{a}_{0,2}^{1}(\left|\eta\right|^{2})\right)\chi_{4}\right],
	\end{aligned}
	\]
	\[
	 e_{2}\left(\eta\right)=\mathscr{L}_{2}\Bigl[\mathtt{a}_{2,1}^{0}(\left|\eta\right|^{2})\chi_{0}+i\mathtt{a}_{2,2}^{1}(\left|\eta\right|^{2})\sum_{j=1}^{3}\eta_{j}\chi_{j}+\mathtt{a}_{2,1}^{4}(\left|\eta\right|^{2})\chi_{4}\Bigr],
	\]
	\[
	e_{3}\left(\eta\right)=\mathscr{L}_{3}\left[\mathtt{a}_{3,1}^{2}(\left|\eta\right|^{2})g^{-1}\cdot\chi_{2}\right],
	\]
	\[
	e_{4}\left(\eta\right)=\mathscr{L}_{4}\left[\mathtt{a}_{3,1}^{2}(\left|\eta\right|^{2})g^{-1}\cdot\chi_{3}\right].
	\]
Here
	\[\mathscr{L}_j=
	 \left[1+\left(L-\mathrm{P}_{1}i\xi\cdot\eta-\varrho_{j}\left(\left|\eta\right|\right)\right)^{-1}\left(\mathrm{P}_{1}i\xi\cdot\eta\right)\right],\]
and $g\in O(3)$, sends $\frac{\eta}{|\eta|}$ to $(1,0,0)^T$ with group action defined in equation \eqref{group-action}.
\end{lemma}
\begin{remark}
	This lemma is a generalization of Lemma 7.8 \cite{[LiuYu1]} and will be proved in section \ref{Pf-smooth}.
\end{remark}

The following combinations of weight functions $w(t,x,\xi)$ and $\varrho(x,\xi)$ are needed for (weighted) energy estimate (Lemmas \ref{improve} and \ref{improve*}, Propositions \ref{weig_1} and  \ref{weig_2}):
\begin{equation}\label{weight-functions}
\left\{
\begin{array}{lll}
w(t,x,\xi)=1,& \varrho(x,\xi)=1 & \mbox{for} \quad \gamma\in[-2,1],\\
w(t,x,\xi)=\exp\left(\frac{\left<x\right>-Mt}{2D}\right), & \varrho(x,\xi)=\exp\left(\frac{\left<x\right>}{D}\right)& \mbox{for}\quad \gamma\in[-1,1],\\
w(t,x,\xi)=\exp\left(\frac{\alpha \vartheta(t,x,\xi)}{2}\right), & \varrho(x,\xi)=\exp\left(\alpha \vartheta(0,x,\xi) \right)& \mbox{for}\quad \gamma\in[-2,-1).\\
\end{array}
\right.
\end{equation}
Here $D>0$ is large, $\alpha>0$ is small, both need to be chosen later. And
\begin{align*}
\vartheta(t,x,\xi)&=5\Big(\de(\left<x\right>-Mt)\Big)^{\frac{2}{1-\ga}}(1-\chi)
\\
&\quad+\bigg[(1-\chi)\de\left(\left<x\right>-Mt\right)\left<\xi\right>^{1+\ga}+3\left<\xi\right>^{2}\bigg]\chi\,.
\end{align*}
The cut-off function $\chi$ is short for
$$
\chi=\chi\Big( \de(\left<x\right>-Mt)\left<\xi\right>^{\ga-1} \Big)\,.
$$
Now, let us state the regularization effect of the linearized Landau equation with weights in small time, this is the key lemma in this paper and
we will prove it in section \ref{Hypo}.
\begin{lemma}\label{improve}
Let $f$ be the solution to equation \eqref{in.1.c}. Let weight function $\varrho(x,\xi)$ be any one of three cases in \eqref{weight-functions}. Then the following regularization estimate holds:
\[
\left\| f(t)\right\|_{H^2_x L^2_\xi(\varrho)} \lesssim t^{-3} \left\| f_0\right\|_{L^2(\varrho\,\mathcal{M})}\quad \mbox{for}\quad 0<t\leq 1,
\]
here
\begin{align*}
\mathcal{M}\equiv \begin{cases}
1 & \ga\in\left[0,1\right],\\
\left<\xi\right>^{2 |\ga|}  & \ga \in \left[-2,0\right).
\end{cases}
\end{align*}
\end{lemma}

\section{Wave inside the finite Mach number region}\label{inside}

In this section, we want to study the solution to the linearized Landau equation inside the finite Mach number region, i.e., the large time behavior of the solution. Using the Fourier transform, the solution to the linearized Landau equation can be written as
\begin{equation}
\displaystyle
\mathbb{G}^{t}f_{0}=f(t,x,\xi)= \int_{\R^{3}}e^{i\eta x+(-i \xi\cdot\eta+L)t}\widehat{f}_{0}(\eta,\xi)d\eta\,,
\end{equation}
where $\widehat{f}$ denotes the Fourier transform with respect to the space variable.
We can decompose the solution $f$ into the long wave part $\mathbb{G}^{t}_{L}f_{0}$ and the short wave part $\mathbb{G}^{t}_{S}f_{0}$:
\begin{equation}\begin{array}{l}\label{bot.2.e}
\displaystyle
\mathbb{G}^{t}_{L}f_{0}= \int_{|\eta|<\de_{0}}e^{i\eta x+(-i \xi\cdot\eta+L)t}\widehat{f}_{0}(\eta,\xi)d\eta\,,
\\ \\
\displaystyle
\mathbb{G}^{t}_{S}f_{0}= \int_{|\eta|>\de_{0}}e^{i\eta x+(-i \xi\cdot\eta+L)t}\widehat{f}_{0}(\eta,\xi)d\eta\,.
\end{array}
\end{equation}
The following long-short wave analysis relies on spectral analysis (Lemma \ref{pr12}).
\begin{lemma}\label{short}{\rm(Short wave $\mathbb{G}^{t}_{S}$)} Let $\ga\geq -2$ and $f_{0}\in L^{2}$, there exists constant $c>0$ such that
\begin{equation}\label{bot.2.f}
\|\mathbb{G}^{t}_{S}f_{0}\|_{L^{2}}\lesssim e^{-ct}\|f_{0}\|_{L^{2}}\,.
\end{equation}
\end{lemma}

In order to study the long wave part $\mathbb{G}_{L}^{t}$, we need to decompose it further into the fluid part and non-fluid part, i.e. $\mathbb{G}_{L}^{t}=\mathbb{G}_{L;0}^{t}+\mathbb{G}_{L;\perp}^{t}$, where
\begin{equation}\begin{array}{l}\label{bot.2.g}
\displaystyle \mathbb{G}^{t}_{L;0}f_{0}= \int_{|\eta|<\de_{0}}\sum_{j=0}^{4}e^{\varrho_{j}(\eta )t}e^{i \eta x}\big<e_{j}(-\eta  ), \hat{f_{0}}\big>_{\xi}e_{j}(\eta  )d\eta\,,
\\ \\
\displaystyle \mathbb{G}^{t}_{L;\perp}f_{0}= \int_{|\eta|<\de_{0}}e^{i \eta x}e^{(-i \xi\cdot\eta +L)t}\Pi_{\eta}^{\perp}\hat{f}_{0}d\eta\,.
\end{array}\end{equation}
For the non-fluid long wave part, it is easy to get the following property:
\begin{proposition}\label{nonflu_long}{\rm(Non-fluid long wave $\mathbb{G}^{t}_{L;\perp}$)} Let $\ga\geq -2$ and $f_{0}\in L^{2}$, there exists a constant $c>0$ such that
\begin{equation}\label{bot.2.h}
\|\mathbb{G}^{t}_{L;\perp}f_{0}\|_{H^{s}_{x}L^{2}_{\xi}}\lesssim e^{-ct}\|f_{0}\|_{L^{2}_{x}L^{2}_{\xi}}
\end{equation}
for any $s\geq0$.
\end{proposition}
By the detailed information of spectrum and eigenfunctions, we are able to estimate the long wave fluid part, which gives the leading order of solution at large time.
\begin{proposition}\label{flu_smooth}{\rm(Fluid Wave $\mathbb{G}_{L;0}^{t}$)}
	Let $\textbf{c}=\sqrt{5/3}$ be the sound speed associated with the normalized
	global Maxwellian. Let $f_0$ be compactly supported in $x$.
	
	(1) For $-1\leq\gamma\leq1$ and any given Mach number $\mathbb{M}>1$,
	there exists positive constant $C$ such that for $\left|x\right|\leq\left(\mathbb{M}+1\right)\textbf{c}t$,
	\[
	\left|\mathbb{G}^{t}_{L;0}f_{0}\right|_{L_{\xi}^{2}}  \leq C
	\left[
	\begin{array}[c]{l}
	\left(1+t\right)^{-2}e^{-\frac{\left(\left|x\right|-\textbf{c}t\right)^{2}}{Ct}}+\left(1+t\right)^{-3/2}e^{-\frac{\left|x\right|^{2}}{Ct}}\\[2mm]
	+\textbf{1}_{\{\left|x\right|\leq \textbf{c}t\}} \left(1+t\right)^{-3/2}\left(1+\frac{\left|x\right|^{2}}{1+t}\right)^{-3/2}+e^{-t/C}
	\end{array}
	\right]
	\left\|f_0\right\|_{L^1_x L^2_\xi}\,.
	\]
	
	(2) For $-2\leq\gamma<-1$ and any given positive integer $N$, there
	exists positive constant C such that
	\[
	\left|\mathbb{G}^{t}_{L;0}f_{0}\right|_{L_{\xi}^{2}} \leq C
	\left[
	\begin{array}[c]{l}
	 \left(1+t\right)^{-2}\left(1+\frac{\left(\left|x\right|-\textbf{c}t\right)^{2}}{1+t}\right)^{-N}+\left(1+t\right)^{-3/2}\left(1+\frac{\left|x\right|^{2}}{1+t}\right)^{-N}\\[2mm]
	+\textbf{1}_{\{\left|x\right|\leq \textbf{c}t\}}\left(1+t\right)^{-3/2}
	\left(1+\frac{\left|x\right|^{2}}{1+t}\right)^{-3/2}+ e^{-t/C}
	\end{array}
	\right]
	\left\|f_0\right\|_{L^1_x L^2_\xi}\,.
	\]
\end{proposition}
\begin{proof}
	When $-1\leq\gamma\leq1$, by Lemma \ref{lem:smooth}, this corresponds to analytic
	case and hence complex analytic technique is applicable. This lemma
	follows from similar calculations as those in sections 7.4-7.5 \cite{[LiuYu1]}.
	As for $-2\leq\gamma<-1$, this corresponds to merely smooth case. Nevertheless, one
	can still use the framework of the above proof, such as Huygens pair, contact pair, rotational pair
and Riesz pair decomposition. But when following the argument, one needs to replace "analytic" by "smooth" and
 complex analytic techniques by real variables techniques accordingly. Here the smooth part of Lemma \ref{lem:smooth}
 is necessary in the calculations. We only list two crucial lemmas which are used for analytic case and smooth case
 respectively in estimate and omit the details.
\begin{lemma}
[{Reformulation of Lemma 7.11, \cite{[LiuYu1]}}]\label{heatkernel_ana} Suppose
that $g(\eta,t,\xi )$ is analytic in $\eta$ for $|\eta|<\delta_{0}\ll1$ and
satisfies
\[
|g(\eta,t,\xi )|_{L^{2}_{\xi }}\lesssim e^{-A|\eta|^{2}t+O(|\eta|^{4})t}\,,
\]
for some $A>0$. Then in the region of $|x|<(\mathfrak{M}+1)t$, $\mathfrak{M}$
is any given positive constant, there exists a constant $C$ such that the
following inequality holds:
\[
\displaystyle \left|  \int_{|\eta|<\delta_{0}} e^{i x\cdot\eta} \eta^{\alpha} g(\eta,t,
\xi )d\eta\right|  _{L^{2}_{\xi }} \leq C\left[  (1+t)^{-\frac{3+|\alpha|}{2}} e^{-\frac
{|x|^{2}}{Ct}} +e^{-t/C} \right]  .
\]

\end{lemma}

\begin{lemma}
[{Reformulation of Lemma 2.2, \cite{[LiuWang]}}]\label{heatkernel_smo} Let
$x,\eta,\xi \in\mathbb{R}^{3}$. Suppose $g(\eta,t,\xi )$ has compact support in the
variable $\eta$, and there exists a constant $b>0$, such that $g(\eta,t,\xi )$
satisfies
\[
\left\vert D_{\eta}^{\beta}(g(\eta,t,\xi ))\right\vert _{L_{\xi }^{2}}\leq
C_{\beta}(1+t^{|\beta|/2})e^{-b|\eta|^{2}t}%
\]
for any multi-index $\beta$ with $|\beta|\leq2N$,
then there exists a positive constant $C_{N}$ such that
\[
\left\vert \int_{|\eta|<\delta_{0}}e^{ix\cdot\eta}g(\eta,t,\xi )d\eta\right\vert
_{L_{\xi }^{2}}\leq C_{N}\left[  (1+t)^{-3/2}B_{N}(|x|,t)+e^{-t/C_N}\right]  ,
\]
where $N$ is any fixed integer and
\[
B_{N}(|x|,t)=\left(  1+\frac{|x|^{2}}{1+t}\right)^{-N}.
\]

\end{lemma}
	From these two lemmas, one can see the origin of difference between heat kernel in analytic case and algebraic decay in smooth case.
\end{proof}

The following lemma is the regularization effect for the linearized Landau equation:
\begin{lemma}\label{improve*}
Let $f$ be the solution to  equation \eqref{in.1.c}. There exists constant $C$ such that for $t\geq 1$
\[
\left\|f(t)\right\|_{H^{2}_x L^2_\xi}\leq C\left\|f_{0} \right\|_{L^2 (\mathcal{M})},
\]
where
\begin{equation}\label{eq_l13.1}
\mathcal{M}\equiv \begin{cases}
1& \ga\in\left[0,1\right],\\
\left<\xi\right>^{2|\ga|}  & \ga \in \left[-2,0\right).
\end{cases}
\end{equation}
\end{lemma}

\begin{proof}
By Lemma \ref{improve} (take the weight function $\varrho=1$), one can improve the regularity of solution in finite time,
\[
\left\|f(1)\right\|_{H^{2}_x L^2_\xi }\leq C\left\|f(0) \right\|_{L^2 (\mathcal{M})}\,.
\]
Then standard energy estimate gives for $t>1$ (for example, cf \cite{[Guo]} Lemma 4),
\[
\left\|f(t)\right\|_{H^{2}_x L^2_\xi}\leq C\left\|f(1)\right\|_{H^{2}_x L^2_\xi }\,.
\]
This completes the proof of the lemma.
\end{proof}

Now, we are in the position to get the pointwise behavior of the short wave part $\mathbb{G}^{t}_{S}f_{0}$:
\begin{proposition}\label{short*} For $t\geq1$,
\[
\|\mathbb{G}^{t}_{S}f_{0}(t)\|_{L^{\infty}_{x}L^{2}_{\xi}}\lesssim e^{-ct}\|f_{0}\|_{L^2(\mathcal{M})}\,.
\]
\end{proposition}
\begin{proof}
Note that
\[
\mathbb{G}^{t}_{S}f_{0}=f(t)-\mathbb{G}^{t}_{L}f_{0}.
\]
By Lemma \ref{short} and \ref{improve*}, we have
\begin{align*}
& \left\|\mathbb{G}^{t}_{S}f_{0} \right\|_{L^2 }\lesssim e^{-C t} \left\|f_{0} \right\|_{L^2},\\
& \left\|\mathbb{G}^{t}_{S}f_{0} \right\|_{H^2_x L^2_\xi }\leq \left\|f(t) \right\|_{H^2_x L^2_\xi}+\left\|\mathbb{G}^{t}_{L}f_{0} \right\|_{H^2_x L^2_\xi }\lesssim \left\|f(0) \right\|_{L^2 (\mathcal{M})}.
\end{align*}
Use Sobolev inequality to conclude that
\[
\|\mathbb{G}^{t}_{S}f_{0}(t)\|_{L^{\infty}_{x}L^{2}_{\xi}}\lesssim e^{-ct}\|f_{0}\|_{L^2(\mathcal{M})}\,.
\]
\end{proof}
With Propositions \ref{nonflu_long}, \ref{flu_smooth} and \ref{short*}, we have the structure of $\mathbb{G}^t f_0$ for $\left<x\right>\leq 2M t$.

\begin{proposition} Let $f(t)=\mathbb{G}^t f_0$ be the solution to the linearized Landau equation \eqref{in.1.c}, and let $\textbf{c}=\sqrt{5/3}$
be the sound speed associated with normalized global Maxwellian. If $t\geq 1$, then
\begin{enumerate}
\item For $-2\leq\gamma<-1$ and any given positive integer $N$, there exists positive constant $C$ (depending on $N$) such that
\[
	\left|\mathbb{G}^{t}f_{0}\right|_{L_{\xi}^{2}}  \leq C\left[
	\begin{array}[c]{l}
		 \left(1+t\right)^{-2}e^{-\frac{\left(\left|x\right|-\textbf{c}t\right)^{2}}{Ct}}+\left(1+t\right)^{-3/2}e^{-\frac{\left|x\right|^{2}}{Ct}}\\[2mm]
		+\textbf{1}_{\{\left|x\right|\leq \textbf{c}t\}} \left(1+t\right)^{-3/2}\left(1+\frac{\left|x\right|^{2}}{1+t}\right)^{-3/2}+e^{-t/C}
	\end{array}
	\right]|||f_0||| \, .
\]
\item For $-1\leq\gamma\leq 1 $, there exists positive constant $C$  such that
\begin{align}
\left| \mathbb{G}^t f_0 \right|_{L^2_\xi} &\leq C\left[
\begin{array}[c]{l}
	 \left(1+t\right)^{-2}\left(1+\frac{\left(\left|x\right|-\textbf{c}t\right)^{2}}{1+t}\right)^{-N}+\left(1+t\right)^{-3/2}\left(1+\frac{\left|x\right|^{2}}{1+t}\right)^{-N}\\[2mm]
	+\textbf{1}_{\{\left|x\right|\leq \textbf{c}t\}}\left(1+t\right)^{-3/2}
	\left(1+\frac{\left|x\right|^{2}}{1+t}\right)^{-3/2}+ e^{-t/C}
\end{array}
\right]|||f_0||| \, .\nonumber
\end{align}
\end{enumerate}
Here $\textbf{1}_{\{\cdot\}}$ is the indicator function,
\[
|||f_0|||\equiv \max \left\{ \left\|f_0\right\|_{L^{2} (\mathcal{M})},\left\|f_0\right\|_{L^1_x L^2_\xi} \right\},
\]
and
\[
\mathcal{M}\equiv \begin{cases}
1 & \ga\in\left[0,1\right],\\
\left<\xi\right>^{2 |\ga|}  & \ga \in \left[-2,0\right).
\end{cases}
\]
\end{proposition}

\section{Wave outside the finite Mach number region}\label{outside}

In this section, we will study the solution outside the finite Mach number region, i.e., the behavior of the solution when $|x|$ large. In order to do this, we need some weighted energy estimates.

\subsection{Weighted energy estimate: $-1\leq \ga\leq 1$}

\begin{proposition}\label{weig_1} Let $-1\leq\ga\leq 1$ and $f$ solves the linearized Landau equation \eqref{in.1.c}. Consider the weight function
\[
w(x,t)=e^{\frac{\left\langle x\right\rangle -Mt}{D}},
\]
where $D$ and $M$ are large number to be choose later. Then we have
$$
\|wf(t)\|_{H^{2}_{x}L^{2}_{\xi}}\lesssim \|wf(1)\|_{H^{2}_{x}L^{2}_{\xi}}\quad \mbox{for}~ t\geq1.
$$
\end{proposition}
\begin{proof}
For simplicity, we only prove the $H^{1}_{x}$ estimate. The $H_{x}^{2}$ estimate is completely the same as $H_{x}^{1}$ estimate. Let $u=wf$, then $u$ solves the equation
$$
\pa_{t}u+\xi\cdot\nabla_{x}u+\frac{1}{D} \Big(M-\frac{x\cdot \xi}{\left<x\right>}\Big)u=Lu\,.
$$
For the $L^2$ estimate, the energy estimate gives
\begin{align*}
\frac{1}{2}\frac{d}{dt}\|u\|^{2}_{L^{2}}+\frac{1}{D}\int_{\R^{6}} \Big(M-\frac{x\cdot \xi}{\left<x\right>}\Big)u^{2}d\xi dx
+\int_{\R^{6}} (-Lu)u d\xi dx  =0\,.
\end{align*}
For the first derivative estimate, note that $\nabla_{x}u$ satisfies the equation
$$
\pa_{t}\left(\nabla_{x}u\right)+\xi\cdot\nabla_{x}\left(\nabla_{x}u\right)+\frac{1}{D} \Big(M-\frac{x\cdot \xi}{\left<x\right>}\Big)\left(\nabla_{x}u\right)=L\left(\nabla_{x}u\right)+\frac{1}{D}\left[\nabla_{x},\frac{x\cdot \xi}{\left<x\right>}\right]u\,,
$$
direct calculation gives
\begin{align*}
\frac{1}{2}\frac{d}{dt}\|\nabla_{x}u\|^{2}_{L^{2}}+\frac{1}{D}\int_{\R^{6}} \Big(M-\frac{x\cdot \xi}{\left<x\right>}\Big)|\nabla_{x}u|^{2}d\xi dx
+\int_{\R^{6}} \left(-L\nabla_{x}u,\nabla_{x}u\right)d\xi dx  \leq \frac{1}{D}\int_{\R^{6}}|\xi||u|^{2}+|\xi||\nabla_{x}u|^{2}d\xi dx  \,.
\end{align*}
By Lemma \ref{lem-a}, we have
$$
\Big|\frac{1}{D}\int_{\R^{6}} |\xi||u|^{2}d\xi dx  \Big|\leq \frac{1}{D}\int_{\R^{6}} \Big[C_{2}u^{2}+C_{1}(-Lu)u\Big]d\xi dx\, ,
$$
and
$$
\Big|\frac{1}{D}\int_{\R^{6}} |\xi||\nabla_{x}u|^{2}d\xi dx  \Big|\leq \frac{1}{D}\int_{\R^{6}} \Big[C_{2}|\nabla_{x}u|^{2}+C_{1}\left(-L\nabla_{x}u,\nabla_{x}u\right)\Big]d\xi dx \,,
$$
this means
\begin{align*}
\frac{1}{2}\frac{d}{dt}\|u\|^{2}_{H^{1}_{x}L^{2}_{\xi}}&+\Big(\frac{M}{D}-\frac{2C_{2}}{D}\Big)\|u\|^{2}_{H^{1}_{x}L^{2}_{\xi}}
+\Big(1-\frac{C_{1}}{D}\Big)\int_{\R^{6}} (-Lu)u+\left(-L\nabla_{x}u,\nabla_{x}u\right)d\xi dx  \leq 0\,.
\end{align*}
This completes the proof of the lemma by choosing $M$, $D$ large enough.
\end{proof}

\subsection{Weighted energy estimate: $-2\leq \ga<-1$} Let us recall the weight function first. We set the exponent of the weight function to be
\begin{align*}
\vartheta(t,x,\xi)&=5\Big(\de(\left<x\right>-Mt)\Big)^{\frac{2}{1-\ga}}(1-\chi)
\\
&\quad+\bigg[(1-\chi)\Big(\de(\left<x\right>-Mt)\Big)\left<\xi\right>^{1+\ga}+3\left<\xi\right>^{2}\bigg]\chi\,,
\end{align*}
where we used the simplified notation
$$
\chi=\chi\Big( \de(\left<x\right>-Mt)\left<\xi\right>^{\ga-1} \Big)\,,
$$
$M$ is a large positive constant, $\de$ is a small positive constant, all of them need to be chosen later. The motivation of the weight function in this case comes from \cite{[cc]}. We define
$$H_{+}=\left\{(x,\xi): \de(\left<x\right>-Mt)>2\left<\xi\right>^{1-\ga} \right\}\,,$$
$$H_{0}=\left\{(x,\xi): \left<\xi\right>^{1-\ga}\leq\de(\left<x\right>-Mt)\leq 2\left<\xi\right>^{1-\ga} \right\}\,,$$
and
$$H_{-}=\left\{(x,\xi): \de(\left<x\right>-Mt)<\left<\xi\right>^{1-\ga} \right\}\,.$$

\begin{proposition}\label{weig_2} Let $-2\leq\ga< -1$ and $f$ solves the linearized Landau equation \eqref{in.1.c}, consider the weight function
\[
w=e^{\alpha\vartheta(t,x,\xi)}.
\]
Then we have
$$
\|wf(t)\|_{H^{2}_{x}L^{2}_{\xi}}\lesssim \|wf(1)\|_{H^{2}_{x}L^{2}_{\xi}}+t\left\|f_{0} \right\|_{L^2 (\mathcal{M})}\quad \mbox{for}~t\geq1.
$$
\end{proposition}
\begin{proof}
Let $u=wf$. Note that $u$ solves the equation
\begin{equation}\label{uu}
\pa_{t}u+\xi\cdot\nabla_{x}u-\alpha (\pa_{t}\vartheta +\xi\cdot\nabla_{x}\vartheta )u-e^{\alpha \vartheta }Le^{-\alpha \vartheta }u=0\,.
\end{equation}
The energy estimate gives
\begin{align*}
\frac{1}{2}\frac{d}{dt}\|u\|^{2}_{L^{2}}
=\int_{\R^{3}}\alpha \big<u, (\pa_{t}\vartheta +\xi\cdot\nabla_{x}\vartheta )u\big>_{\xi}dx
+\int_{\R^{3}} \big<u, e^{\alpha \vartheta }Le^{-\alpha \vartheta }u\big>_{\xi}dx\,.
\end{align*}
Note that $L=\widetilde{\Lambda}+\widetilde{K}$, for simplicity, define $\widetilde{K}_{w}f=e^{\alpha \vartheta }K(e^{-\alpha \vartheta }f)$, then
\begin{align}\label{L-Lw}
\big<u, e^{\alpha \vartheta }Le^{-\alpha \vartheta }u\big>_{\xi}&=\big<u, e^{\alpha \vartheta }\widetilde{\Lambda}e^{-\alpha \vartheta }u\big>_{\xi}
+\big<u, \widetilde{K}_{w}u\big>_{\xi}\nonumber
\\
&=\big<u, Lu\big>_{\xi}-\big<u, (\widetilde{K}-\widetilde{K}_{w})u\big>_{\xi}+\alpha^{2}\big<u^{2}\left(\nabla_{\xi}\vartheta , \sigma\nabla_{\xi}\vartheta\right) \big>_{\xi}\,.
\end{align}
For the first term of (\ref{L-Lw}), by Lemma \ref{prop1} \eqref{coercivity}, we have
$$
\int_{\R^{3}}\big<u, Lu\big>_{\xi}dx\leq- \|\mathrm{P}_{1}u\|^{2}_{L^{2}_{\sigma}}\,.
$$
For the second term of (\ref{L-Lw}), define
$$
A_{++}=\left\{(x,\xi,\xi_{*}): (x,\xi)\in H_{+}, (x,\xi_{*})\in H_{+}\right\}\,,
$$
$$
A_{**}=\left\{(x,\xi,\xi_{*}): (x,\xi)\in H_{0}\cup H_{-}, (x,\xi_{*})\in H_{0}\cup H_{-}\right\}\,,
$$
$$
A_{+*}=\left\{(x,\xi,\xi_{*}): (x,\xi)\in H_{+}, (x,\xi_{*})\in H_{0}\cup H_{-}\right\}\,,
$$
and
$$
 A_{*+}=\left\{(x,\xi,\xi_{*}): (x,\xi)\in H_{0}\cup H_{-}, (x,\xi_{*})\in H_{+}\right\}\,.
$$
Note that
\begin{align}\label{estimate-K}
&\quad\int_{\R^{3}}\big<u, (\widetilde{K}-\widetilde{K}_{w})u\big>_{\xi}dx\nonumber\\
&= \int_{\R^{3\times 3}}u(t,x,\xi)
\tilde{k}(\xi,\xi_{*})\left(1-e^{\alpha(\vartheta(t,x,\xi)-\vartheta(t,x,\xi_{*}))}\right)u(t,x,\xi_{*})d\xi d\xi_{*} dx\\
&= \int_{A_{++}}+\int_{A_{**}}+\int_{A_{+*}\cup A_{*+}}u(t,x,\xi)
\tilde{k}(\xi,\xi_{*})\left(1-e^{\alpha(\vartheta(t,x,\xi)-\vartheta(t,x,\xi_{*}))}\right)u(t,x,\xi_{*})d\xi d\xi_{*} dx\,,\nonumber
\end{align}
it is easy to see that $\vartheta(t,x,\xi)=\vartheta(t,x,\xi_{*})$ in $A_{++}$, this means no contribution in this region. For $A_{**}$,
if $\alpha$ small enough, by the Taylor expansion, we have
$$
\left|1-e^{\alpha(\vartheta(t,x,\xi)-\vartheta(t,x,\xi_{*}))}\right|\lesssim \alpha \left(|\xi|^{2}+|\xi_{*}|^{2}\right)
e^{\tilde{\alpha}(\vartheta(t,x,\xi)-\vartheta(t,x,\xi_{*}))}\,,
$$
where $\tilde{\alpha}\in (0,\alpha)$, this implies
\[
\begin{aligned}
&\quad\left|\int_{A_{**}}u(t,x,\xi)
\tilde{k}(\xi,\xi_{*})\left(1-e^{\alpha(\vartheta(t,x,\xi)-\vartheta(t,x,\xi_{*}))}\right)u(t,x,\xi_{*})d\xi d\xi_{*}dx\right|\\
&\lesssim  \alpha \int_{A_{**}} e^{-c\left(|\xi|^2+|\xi_{*}|^2\right)} \left|\xi-\xi_{*}\right|^{\gamma} \left|u(t,x,\xi)\right|\left|u(t,x,\xi_{*})\right|d\xi d\xi_{*} dx\\
&\lesssim \alpha\int_{A_{**}\cap\{|\xi-\xi_{*}|\leq 1\}} + \alpha\int_{A_{**}\cap\{|\xi-\xi_{*}|> 1\}}\\
&=: T_1+T_2.
\end{aligned}
\]
For $T_1$,
\[
\begin{aligned}
T_1 &=\alpha\int_{\R^{3}} dx \int_{\R^{3}} d\xi  \,\textbf{1}_{ A_{**}} e^{-c|\xi|^2} \left|u(t,x,\xi)\right| \left(\int_{\R^{3}} d\xi_{*}\,\textbf{1}_{\{|\xi-\xi_{*}|\leq1\}} |\xi-\xi_{*}|^{\gamma} e^{-c|\xi_{*}|^2} \left|u(t,x,\xi_{*})\right| \right)\\
&\lesssim \alpha \int_{H_{0}\cup H_{-}} e^{-c|\xi|^2} \left|u(t,x,\xi)\right|  \textbf{M}(e^{-c|\xi|^2}|u(t,x,\xi)| )d\xi dx,\\
\end{aligned}
\]
where $\textbf{M}$ denotes the maximal function. The basic property of maximal function states that $\textbf{M}$ is a bounded operator on $L^2$. Then by Cauchy-Schwarz inequality,
\[
T_1\lesssim \alpha \int_{H_{0}\,\cup H_{-}} e^{-c|\xi|^2} |u|^2 d\xi dx.
\]
It also follows from  Cauchy-Schwarz inequality that
\[
\begin{aligned}
T_2 & \lesssim \alpha \int_{\R^{3}} \left(\int_{\R^{3}}  e^{-c|\xi|^2} |u(t,x,\xi)|^2 d\xi\right)^{2}\, \textbf{1}_{H_{0}\,\cup H_{-}}\, dx\\
    & \lesssim \alpha \int_{H_{0}\,\cup H_{-}} e^{-c|\xi|^2} |u|^2 d\xi dx.
\end{aligned}
\]
In the region $A_{+*}\cup A_{*+}$, by symmetry and the Taylor expansion, there exists $\widehat{a}\in (0,a)$ such that
\begin{align*}
&\quad\left|\int_{A_{*+}\cup A_{+*}}u(t,x,\xi)
\tilde{k}(\xi,\xi_{*})\left(1-e^{\alpha(\vartheta(t,x,\xi)-\vartheta(t,x,\xi_{*}))}\right)
u(t,x,\xi_{*})d\xi d\xi_{*}dx\right|\\
&\lesssim \alpha^{2}\int_{A_{+*}}\left|u(t,x,\xi)
\tilde{k}(\xi,\xi_{*})\left(\vartheta^{2}(t,x,\xi)+\vartheta^{2}(t,x,\xi_{*})\right)e^{\widehat{\alpha}(\vartheta(t,x,\xi)-\vartheta(t,x,\xi_{*}))}
u(t,x,\xi_{*})\right|d\xi d\xi_{*}dx\\
&\lesssim\alpha^{2} \int_{H_{0}\cup H_{-}}e^{-c|\xi_{*}|^{2}}|u|^{2}d\xi_{*} dx
+\alpha^{2}\int_{H_{+}}\Big(\de(\left<x\right>-Mt)\Big)^{\frac{1+\ga}{1-\ga}}e^{-c|\xi|^{2}}|u|^{2}d\xi dx\,,
\end{align*}
for some small constant $c>0$. Here the treatment of singular part in $\tilde{k}(\xi,\xi_{*})$ is similar as for $T_1$, so we omit them. Summing up the above calculations, we have
\begin{align}\label{K-Kw}
&\quad\int_{\R^{3}}\big<u, (\widetilde{K}-\widetilde{K}_{w})u\big>_{\xi}dx\\
&\lesssim (\alpha+\alpha^{2})
 \int_{H_{0}\cup H_{-}}e^{-c|\xi|^{2}}|u|^{2}d\xi dx
+\alpha^{2}\int_{H_{+}}\Big(\de(\left<x\right>-Mt)\Big)^{\frac{1+\ga}{1-\ga}}e^{-c|\xi|^{2}}|u|^{2}d\xi dx\,.\nonumber
\end{align}
Finally, direct calculation gives
\begin{align*}
\nabla_{\xi}\vartheta &  =\left[  (\gamma-1)(1-2\chi)  \left(\de(\left\langle
x\right\rangle -Mt)\right)  \left\langle \xi\right\rangle ^{\gamma+1}+3(\gamma-1)|\xi|^{2}-5(\gamma-1)\left(
\de(\left\langle x\right\rangle -Mt)\right)  ^{\frac{2}{1-\gamma}%
}\right] \\
&  \quad\times\left[    \de(\left\langle x\right\rangle -Mt)
\left\langle \xi\right\rangle ^{\gamma-2}\right]  \frac{\xi}{\left\langle
\xi\right\rangle }\chi^{\prime}\\
&  \quad+\left[  \gamma  \left(\de(\left\langle x\right\rangle
-Mt)\right)  \left\langle \xi\right\rangle ^{\gamma}\right]  \frac{\xi}{\left\langle \xi\right\rangle
}(1-\chi)\chi+6\xi\chi\,.
\end{align*}
This implies
\[
\nabla_{\xi}\vartheta \sim\left\langle \xi\right\rangle\frac{\xi}{\left<\xi\right>}\quad
\hbox{on}\quad H_{0}\,\cup H_{-} \,,
\]
and
\[
\nabla_{\xi}\vartheta=0\quad\hbox{on}\quad H_{+}\,,
\]
hence
$$
|\left(\nabla_{\xi}\vartheta , \sigma\nabla_{\xi}\vartheta\right) |\lesssim \left<\xi\right>^{\ga+2}\quad
\hbox{on}\quad H_{0}\,\cup H_{-} \,.
$$
We have the estimate of the third term of (\ref{L-Lw})
\begin{align*}
\alpha^{2}\big|\int_{\R^{3}}\big<u^{2}\left(\nabla_{\xi}\vartheta , \sigma\nabla_{\xi}\vartheta\right) \big>_{\xi}dx\big|
&\leq \alpha^{2} \|\mathrm{P}_{1}u\|^{2}_{L^{2}_{\sigma}}+\alpha^{2}\int_{H_{0}\,\cup H_{-}}\left<\xi\right>^{\ga+2}|\mathrm{P}_{0}u|^{2}d\xi dx
\\
 &\leq \alpha^{2} \|\mathrm{P}_{1}u\|^{2}_{L^{2}_{\sigma}}+\alpha^{2}\int_{H_{0}\,\cup H_{-}}|\mathrm{P}_{0}u|^{2} d\xi dx\,.
\end{align*}
One can easily check that
\begin{equation}\label{theta-dec}
\begin{aligned}
\pa_{t}\vartheta
 & =-\delta M\left\langle \xi\right\rangle ^{\gamma+1}\left(  \frac{10
}{1-\gamma}\left[   \left( \delta(\left\langle x\right\rangle -Mt)\right)
\left\langle \xi\right\rangle ^{\gamma-1}\right]  ^{\frac{1+\gamma}{1-\gamma}%
}\left(  1-\chi\right)  +\chi(1-\chi)\right) \\
&  \quad+\de M  \left(
5\left[ \left( \delta(\left\langle x\right\rangle -Mt)\right)
\left\langle \xi\right\rangle ^{\gamma-1}\right]  ^{\frac{2}{1-\gamma}%
}-(1-2\chi)\left[  \left(\delta(\left\langle x\right\rangle -Mt)\right)
\left\langle \xi\right\rangle ^{\gamma-1}\right]-3  \right)  \left\langle
\xi\right\rangle ^{\gamma+1}\chi^{\prime}\,\\
&\leq0,
\end{aligned}
\end{equation}
(the constants $5$ and $3$ are chosen intentionally such that the quantity in
the latter bracket is nonnegative on $H_{0}$) and,
\begin{align*}
\nabla_{x}\vartheta &  =\delta\left(  \nabla_{x}\left\langle x\right\rangle
\right)  \left\langle \xi\right\rangle ^{\gamma+1}\left(  \frac{10
}{1-\gamma}\left[   \left( \delta(\left\langle x\right\rangle -Mt)\right)
\left\langle \xi\right\rangle ^{\gamma-1}\right]  ^{\frac{1+\gamma}{1-\gamma}%
}\left(  1-\chi\right)  +\chi(1-\chi)\right) \\
&  \quad-\de\left(  \nabla_{x}\left\langle x\right\rangle \right)  \left(
5\left[ \left( \delta(\left\langle x\right\rangle -Mt)\right)
\left\langle \xi\right\rangle ^{\gamma-1}\right]  ^{\frac{2}{1-\gamma}%
}-(1-2\chi)\left[  \left(\delta(\left\langle x\right\rangle -Mt)\right)
\left\langle \xi\right\rangle ^{\gamma-1}\right]-3  \right)  \left\langle
\xi\right\rangle ^{\gamma+1}\chi^{\prime}\,,
\end{align*}
hence
\[
\pa_{t}\vartheta=\xi\cdot\nabla_{x}\vartheta=0\quad\hbox{on}\quad H_{-}\,,
\]
for $H_{0},$ we get
\[
|\pa_{t}\vartheta|\lesssim\de M\left\langle \xi\right\rangle ^{\gamma+1}%
\quad\hbox{and}\quad|\xi\cdot\nabla_{x}\vartheta|\lesssim\de\left\langle
\xi\right\rangle ^{\gamma+2}.
\]
Finally, on $H_{+},$ we have
\[
\pa_{t}\vartheta=-\frac{10}{1-\gamma}\delta M\left[  \delta
(\left\langle x\right\rangle -Mt)  \right]  ^{\frac{1+\gamma
}{1-\gamma}},
\]%
\[
\xi\cdot\nabla_{x}\vartheta=\frac{10}{1-\gamma}\delta\frac{\xi\cdot x}%
{\left\langle x\right\rangle }\left[  \delta(\left\langle
x\right\rangle -Mt)  \right]  ^{\frac{1+\gamma}{1-\gamma}}.
\]

Direct calculation gives
\begin{align*}
\alpha\big|\int_{\R^{3}} \big<u, \xi\cdot\nabla_{x}\vartheta u\big>_{\xi}dx\big|&\leq \alpha\de \|\left<\xi\right>^{\frac{\ga+2}{2}}\mathrm{P}_{1}u\|^{2}_{L^{2}}
\\
&\quad+\alpha\de\int_{H_{+}}\Big(\de(\left<x\right>-Mt)\Big)^{\frac{1+\ga}{1-\ga}}\left<\xi\right> |\mathrm{P}_{0}u|^{2}d\xi dx
+\alpha\de\int_{H_{0}}	\left<\xi\right>^{\gamma+2}|\mathrm{P}_{0}u|^{2}d\xi dx
\end{align*}
and
\begin{align*}
\alpha\int_{\R^{3}} \big<u, \pa_{t}\vartheta u\big>_{\xi}dx&\leq \alpha \de M\|\left<\xi\right>^{\frac{\ga+1}{2}}\mathrm{P}_{1}u\|^{2}_{L^{2}}
\\
&\quad-\alpha\de M\int_{H_{+}}\Big(\de(\left<x\right>-Mt)\Big)^{\frac{1+\ga}{1-\ga}}|\mathrm{P}_{0}u|^{2}d\xi dx
+\alpha\de M\int_{H_{0}}|\mathrm{P}_{0}u|^{2}d\xi dx \,.
\end{align*}
In conclusion, we get
\begin{align*}
\frac{d}{dt}\|u\|^{2}_{L^{2}}&\leq  -(C-\alpha-\alpha^{2}-\alpha\de-\alpha\de M) \|\mathrm{P}_{1}u\|^{2}_{L^{2}_{\sigma}}\\
&\quad-\alpha\de M\int_{H_{+}}\Big(\de(\left<x\right>-Mt)\Big)^{\frac{1+\ga}{1-\ga}}|\mathrm{P}_{0}u|^{2}d\xi dx  \\
&\quad+  \alpha(\delta+\alpha) \int_{H_{+}}\Big(\de(\left<x\right>-Mt)\Big)^{\frac{1+\ga}{1-\ga}}\left<\xi\right>|\mathrm{P}_{0}u|^{2}d\xi dx\\
&\quad+\alpha(1+\alpha+\de+\de M)
\int_{H_{0}\cup H_{-}}\left<\xi\right>^{\gamma+2}|\mathrm{P}_{0}u|^{2}d\xi dx\\
&\leq \alpha(1+\alpha+\de+\de M)\|u\|_{L^{2}}\|f\|_{L^{2}}\,,
\end{align*}
if we choose $\de$ small, $\alpha$ much smaller than $\de$ and $M$ large enough.

For the $x$-derivative estimate, we can rewrite (\ref{uu}) as
\begin{align*}
\pa_{t}u+\xi\cdot\nabla_{x}u&=Lu+\left[\alpha \left(\partial_t \vartheta+\xi\cdot\nabla_{x} \vartheta - \nabla_\xi\cdot \left[\sigma\nabla_\xi \vartheta\right]\right)
+ \alpha^2 (\sigma\nabla_\xi \vartheta,\nabla_\xi \vartheta)\right]u\\
&\quad-2\alpha \lambda_1(\xi)\nabla_\xi \vartheta\cdot\nabla_\xi u+\left(\widetilde{K}_{w}-\widetilde{K}\right)u\,.
\end{align*}
We only need to control the commutator
terms:
\begin{equation}
\int_{{\mathbb{R}}^{3}}\left\langle \pa_{x_{i}}u,\left[\alpha\pa_{x_{i}} \left(\partial_t \vartheta+\xi\cdot\nabla_{x} \vartheta
- \nabla_\xi\cdot \left[\sigma\nabla_\xi \vartheta\right]\right)
+\alpha^2\pa_{x_{i}} (\sigma\nabla_\xi \vartheta,\nabla_\xi \vartheta)\right]  u\right\rangle _{\xi }dx, \label{comm1}%
\end{equation}%
\begin{equation}
\int_{{\mathbb{R}}^{3}}\left\langle \pa_{x_{i}}u,-2\alpha \lambda_1(\xi)\nabla_\xi \pa_{x_{i}}\vartheta\cdot\nabla_\xi u
\right\rangle _{\xi }dx\,, \label{comm2}%
\end{equation}
and%
\begin{equation}
\int_{{\mathbb{R}}^{3}}\left\langle \pa_{x_{i}}u,\left(\pa_{x_{i}}\widetilde{K}_{w}\right)u  \right\rangle _{\xi }dx. \label{comm3}%
\end{equation}
It is obvious that the decays of $\pa_{x_{i}}\left[\alpha \left(\partial_t \vartheta+\xi\cdot\nabla_{x} \vartheta
- \nabla_\xi\cdot \left[\sigma\nabla_\xi \vartheta\right]\right)
+ \alpha^2 (\sigma\nabla_\xi \vartheta,\nabla_\xi \vartheta)\right]$ and $\lambda_1(\xi)\nabla_\xi \pa_{x_{i}}\vartheta$ are faster than $\left[\alpha \left(\partial_t \vartheta+\xi\cdot\nabla_{x} \vartheta
- \nabla_\xi\cdot \left[\sigma\nabla_\xi \vartheta\right]\right)
+ \alpha^2 (\sigma\nabla_\xi \vartheta,\nabla_\xi \vartheta)\right]$ and  $\lambda_1(\xi)\nabla_\xi\vartheta$ respectively,
 hence the first term (\ref{comm1}) and the second term (\ref{comm2}) are easy to control. For
the third term (\ref{comm3}), direct calculation gives (similar to the estimate of (\ref{estimate-K}))
\begin{align*}
\left|\int_{{\mathbb{R}}^{3}}\left\langle \pa_{x_{i}}u,\left(\pa_{x_{i}}\widetilde{K}_{w}\right)u  \right\rangle _{\xi }dx\right|&\lesssim
\alpha\de
 \int_{H_{0}\cup H_{-}}e^{-c|\xi|^{2}}|u||\pa_{x_{i}}u|d\xi dx\\
&\quad+\alpha\de\int_{H_{+}}\Big(\de(\left<x\right>-Mt)\Big)^{\frac{1+\ga}{1-\ga}}e^{-c|\xi|^{2}}|u||\pa_{x_{i}}u|d\xi dx
\end{align*}
for some small constant $c>0$, hence the third term can also be controlled. The second derivative estimate is similar and
hence we omit the details. We then have
\begin{align*}
\frac{d}{dt}\Vert u\Vert_{H_{x}^{2}L_{\xi }^{2}}^{2}  \lesssim\Vert
u\Vert_{H_{x}^{2}L_{\xi }^{2}}\Vert f\Vert_{H_{x}^{2}L_{\xi }^{2}}\,.
\end{align*}
By Lemma \ref{improve*}, we get
\begin{align*}
\Vert u(t)\Vert_{H_{x}^{2}L_{\xi }^{2}}  &\lesssim\Vert
u(1)\Vert_{H_{x}^{2}L_{\xi }^{2}}+\int_{1}^{t}\Vert f(s)\Vert_{H_{x}^{2}L_{\xi }^{2}}ds\\
&\lesssim \Vert
u(1)\Vert_{H_{x}^{2}L_{\xi }^{2}}+t \Vert f_{0}\Vert_{L^{2}(\mathcal{M})}\,.
\end{align*}
This completes the proof of the proposition.
\end{proof}

\subsection{Conclusion}

\begin{lemma}\label{w-to-mu}
	 Let $w(t,x,\xi)$ and $\varrho(x,\xi)$ be weight functions defined in \eqref{weight-functions}. Then
	\[
	\left\|wf(t) \right\|_{H^2_x L^2_\xi}\lesssim \left\|f(t) \right\|_{H^2_x L^2_\xi (\varrho)}
	\]
\end{lemma}
\begin{proof}
	It is easy to see that
	\[
	\begin{aligned}
	\left\|wf(t) \right\|_{H^2_x L^2_\xi}^2 & \lesssim \int_{\R^{6}}
f^2 \Big(w^2+ \left|\nabla_x w\right|^2 + \left|D_x^2 w\right|^2\Big)d\xi dx + \int_{\R^{6}} \left|\nabla f\right|^2 \Big(w^2 + \left|\nabla_x w\right|^2\Big)d\xi dx\\
	& \quad + \int_{\R^{6}} \left|D_x^2 f\right|^2 w^2 d\xi dx.
	\end{aligned}
	\]
	For $\gamma\in[-1,1]$, direct computation shows that
	\[
	\left|\nabla_x w\right|,\left|D_x^2 w\right|\lesssim \frac{1}{D} w\,.
	\]
	For $\gamma\in[-2,-1)$, we have
	\[
	\left|\nabla_x w\right|, \left|D_x^2 w\right|\lesssim \alpha\delta^2 w\left( \alpha \left<\xi\right>^{2(1+\gamma)}+\left<\xi\right>^{2\gamma} \right).
	\]
	It then follows that
	\[
	\left\|wf(t) \right\|_{H^2_x L^2_\xi}^2 \lesssim \int_{\R^{6}} \left(f^2 + \left|\nabla_x f\right|^2 +\left|D_x^2 f \right|^2\right)w^2 d\xi dx.
	\]
	On the other hand, from \eqref{theta-dec}, $\vartheta(t,x,\xi)$ is non-increasing in $t$, thus $w(t,x,\xi)\leq w(0,x,\xi)$. By definition,  $\varrho(x,\xi)=w(0,x,\xi)^2$, therefore we conclude
	\[
	\left\|wf(t) \right\|_{H^2_x L^2_\xi}^2 \lesssim \int_{\R^{6}} \left(f^2 + \left|\nabla_x f\right|^2 +\left|D_x^2 f \right|^2\right)\varrho\, d\xi dx=\left\|f(t) \right\|^2_{H^2_xL^2_\xi(\varrho)}.
	\]
\end{proof}

With Propositions \ref{improve}, \ref{weig_1}, \ref{weig_2} and Lemmas \ref{improve*}, \ref{w-to-mu} we have the structure of $\mathbb{G}^t f_0$ for $\left<x\right>> 2M t$.
\begin{proposition}
Let $f$ be the solution to the linearized Landau equation \eqref{in.1.c}, there exists positive constant $M$ such that for $\left<x\right>>2 M t$ and $t\geq 1$:
\begin{enumerate}
\item For $-1\leq\gamma\leq 1 $, there exist positive constants $C$ and $c$ such that
\[
\left|f (t,x) \right|_{L^2_\xi}  \leq C e^{-c\left(\left<x\right>+t\right)} \left\|f_0\right\|_{L^{2} (\mathcal{M})}\,.
\]

\item For $-2\leq\gamma<-1$, for any $\alpha>0$ sufficiently small, there exist positive constants $C$ and $c_\alpha$ such that
\[
 \left| f(t,x) \right|_{L^2_\xi}\leq C e^{-c_\alpha(\left<x\right>+t)^{\frac{2}{1-\gamma}}}
 \left\|f_0\right\|_{L^{2} (e^{\alpha\left|\xi\right|^2}\mathcal{M})}\,.
\]
\end{enumerate}
Here
\[
\mathcal{M}\equiv \begin{cases}
1 & \ga\in\left[0,1\right],\\
\left<\xi\right>^{2|\ga|}  & \ga \in \left[-2,0\right).
\end{cases}
\]
\end{proposition}
\begin{proof}
For the first part, notice that $w(x,t)\geq e^{\frac{\left\langle x\right\rangle +2Mt}{8D}}$ when $\left<x\right>>2Mt$. It follows from Sobolev's inequality and Proposition \ref{weig_1} (replacing $D$ by $2D$) that there exist $C,c>0$ such that
\[
\left|f(t,x)\right|_{L^2_\xi}  \leq C e^{-c\left(\left<x\right>+t\right)}\left\| wf(1)\right\|_{H^2_x L^2_\xi}\quad \mbox{for}~\left<x\right> > 2Mt.
\]
From Lemmas \ref{w-to-mu} and \ref{improve*},
\[
\left\| wf(1)\right\|_{H^2_x L^2_\xi}\lesssim \left\| f(1) \right\|_{H^2_x L^2_\xi (\varrho)}\lesssim \left\|f_0 \right\|_{L^2(\varrho\,\mathcal{M})}\approx \left\|f_0 \right\|_{L^2(\mathcal{M})} .
\]
The last relation is due to  $f_{0}$ has compact support in the $x$-variable. Thus,
\[
\left|f (t,x) \right|_{L^2_\xi}  \leq C e^{-c\left(\left<x\right>+t\right)} \left\|f_0\right\|_{L^{2} (\mathcal{M})}.
\]

For the second part, let $w$ be defined in \eqref{weight-functions}. Applying Proposition \ref{weig_2} with replacing $\alpha$ by $\alpha/2$, we have
\[
\left\| wf(t)\right\|_{H^2_x L^2_\xi} \lesssim \left\| wf(1) \right\|_{H^2_x L^2_\xi}+t\left\|f_{0} \right\|_{L^2 (\mathcal{M})}.
\]
On the other hand, it follows from Lemmas \ref{improve}  and \ref{w-to-mu} that
\[
 \left\| wf(1) \right\|_{H^2_x L^2_\xi}\lesssim \left\| f(1) \right\|_{H^2_x L^2_\xi(\varrho)}\lesssim \left\| f_0 \right\|_{ L^2(\varrho\,\mathcal{M})}\lesssim\left\|f_0\right\|_{L^{2} (e^{3\alpha\left|\xi\right|^2}\mathcal{M})}.
\]
Observe that for $\left<x\right>>2Mt$,
\[
\vartheta(t,x,\xi)\gtrsim \left(\delta (\left<x\right>-Mt)\right)^{\frac{2}{1-\gamma}}.
\]
It follows from Sobolev inequality that
\[
\begin{aligned}
\sup_{(x,t)\in \mathbb{R}^3\times \mathbb{R}_+} e^{\alpha \left(\delta (\left<\left<x\right>\right>-Mt)\right)^{\frac{2}{1-\gamma}}}
\left| f(t,x) \right|_{L^2_\xi}\lesssim (1+t)\left\|f_0\right\|_{L^{2} (e^{3\alpha\left|\xi\right|^2}\mathcal{M})}.
\end{aligned}
\]
Note that for $\left<x\right>>2Mt$,
\[
\left<x\right>-Mt>\frac{\left<x\right>}{3} + \frac{Mt}{3},
\]
therefore there exist positive constants $C$ and $c_\alpha$ such that
\[
 \left| f(t,x) \right|_{L^2_\xi}\leq C (1+t)e^{-c_\alpha(\left<x\right>+t)^{\frac{2}{1-\gamma}}} \left\|f_0\right\|_{L^{2} (e^{3\alpha\left|\xi\right|^2}\mathcal{M})}\,.
\]
Here $\alpha>0$ can be chosen as small as we want.
\end{proof}

\section{Proof of Lemma \ref{lem:smooth}}\label{Pf-smooth}

The goal of this section is to prove Lemma \ref{lem:smooth}, which shows that how the eigenvalues and eigenfunctions of the operator $-i\xi\cdot\eta+L$ depend on Fourier variable $\eta$ (namely, smoothly or analytically). It is a basis for analyzing fluid structure of the solution. Consider eigen-problem
\begin{equation}
\left(-i\xi\cdot\eta+L\right)e_{j}\left(\eta\right)=\varrho_{j}\left(\eta\right)e_{j}\left(\eta\right).\label{eq:eigenprob}
\end{equation}
Here the parameter $\eta\in\mathbb{R}^3$ is three dimensional. Like eigen-problem for the Boltzmann equation \cite{[LiuYu2],[LiuYu1]}, to simplify the analysis, the first step is to reduce the parameter space by some symmetry properties of operator $L$. Thus, we start with the following Lemma.

The orthogonal group $O\left(3\right)$ has a natural action on $L_{\xi}^{2}$.
Let $g\in O\left(3\right)$, $f\in L_{\xi}^{2}$, the action is given
by
\begin{equation}
\label{group-action}
\left(g\cdot f\right)\left(\xi\right)=f\left(g^{-1}\xi\right).
\end{equation}

\begin{lemma}
	\label{lem:commute}The $O\left(3\right)$-action commutes with $L$,
	$\mathrm{P}_{0}$ and $\mathrm{P}_{1}$.
\end{lemma}
\begin{proof}
	The proof are simple calculations. For the nonlinear collision operator,
	we have
	\begin{align*}
		Q\left(g\cdot F,g\cdot F\right)\left(\xi\right) & =\nabla_{\xi}\cdot\int_{\mathbb{R}^{3}}\Phi\left(\xi-\xi_{*}\right)\left[\left(g\cdot F\right)\left(\xi_{*}\right)\nabla_{\xi}\left(g\cdot F\right)\left(\xi\right)-\left(g\cdot F\right)\left(\xi\right)\nabla_{\xi_{*}}\left(g\cdot F\right)\left(\xi_{*}\right)\right]d\xi_{*}\\
		& =\nabla_{\xi}\cdot\int_{\mathbb{R}^{3}}\Phi\left(\xi-\xi_{*}\right)\left[F\left(g^{-1}\xi_{*}\right)\nabla_{\xi}F\left(g^{-1}\xi\right)-F\left(g^{-1}\xi\right)\nabla_{\xi_{*}}F\left(g^{-1}\xi_{*}\right)\right]d\xi_{*}.
	\end{align*}
	Using change of variables $\zeta=g^{-1}\xi$, $\zeta_{*}=g^{-1}\xi_{*}$,
	the gradient $\nabla_{\xi}$ changes to $g\nabla_{\zeta}$, here we
	treat $\nabla_{\zeta}$ as a column vector and $g$ as a $3\times3$
	matrix. It is easy to see $\Phi\left(g\zeta\right)=g\Phi\left(\zeta\right)g^{T}$.
	Thus
	\begin{align*}
		Q\left(g\cdot F,g\cdot F\right)\left(\xi\right) & =g\nabla_{\zeta}\cdot\int_{\mathbb{R}^{3}}g\Phi\left(\zeta-\zeta_{*}\right)g^{T}\left[F\left(\zeta_{*}\right)g\nabla_{\zeta}F\left(\zeta\right)-F\left(\zeta\right)g\nabla_{\zeta_{*}}F\left(\zeta_{*}\right)\right]d\zeta_{*}\\
		& =\nabla_{\zeta}\cdot\int_{\mathbb{R}^{3}}\Phi\left(\zeta-\zeta_{*}\right)\left[F\left(\zeta_{*}\right)\nabla_{\zeta}F\left(\zeta\right)-F\left(\zeta\right)\nabla_{\zeta_{*}}F\left(\zeta_{*}\right)\right]d\zeta_{*}\\
		& =Q\left(F,F\right)\left(\zeta\right)=Q\left(F,F\right)\left(g^{-1}\xi\right)=g\cdot Q\left(F,F\right)\left(\xi\right).
	\end{align*}
	Noting that $Lf=2\mu^{-1/2}Q\left(\mu,\mu^{1/2}f\right)$ and $\mu=\left(2\pi\right)^{-3/2}e^{-\left|\xi\right|^{2}/2}$
	is invariant under $g$ action, we have $g\cdot Lf=Lg\cdot f.$
	
	Let $\left\{ \chi_{i},\,i=0,\dots,4\right\} $ be an orthonormal basis
	of $\mathrm{Ker}\,L$, then $\left\{ g\cdot\chi_{i},\,i=0,\dots,4\right\} $
	is also an orthonormal basis. Hence
	\begin{align*}
		\mathrm{P}_{0}\left(g\cdot f\right)\left(\xi\right) & =\sum_{j=0}^{4}\int_{\mathbb{R}^{3}}f\left(g^{-1}\zeta\right)\chi_{j}\left(\zeta\right)d\zeta\,\chi_{j}\left(\xi\right)\\
		& =\sum_{j=0}^{4}\int_{\mathbb{R}^{3}}f\left(\zeta\right)\chi_{j}\left(g\zeta\right)d\zeta\,\chi_{j}\left(\xi\right)\\
		& =\sum_{j=0}^{4}\int_{\mathbb{R}^{3}}f\left(\zeta\right)\left(g^{-1}\cdot\chi_{j}\right)\left(\zeta\right)d\zeta\,\left(g^{-1}\cdot\chi_{j}\right)\left(g^{-1}\xi\right)\\
		& =\left(\mathrm{P}_{0}f\right)\left(g^{-1}\xi\right)=\left(g\cdot\mathrm{P}_{0}f\right)\left(\xi\right).
	\end{align*}
	This in turn shows $\mathrm{P}_{1}=\mathrm{Id}-\mathrm{P}_{0}$ also
	commutes with $g$ action. The lemma is thus proved.
\end{proof}

As mentioned before, we can simplify the eigenvalue problem with this lemma. We choose a special group element $g\in O\left(3\right)$ which sends $\frac{\eta}{\left|\eta\right|}$
to $\left(1,0,0\right)^{T}$. Applying $g$ to (\ref{eq:eigenprob}),
we have
\begin{align*}
	\left(-ig^{-1}\xi\cdot\eta+L\right)\left(g\cdot e_{j}\right) & =\left(-i\xi\cdot g\eta+L\right)\left(g\cdot e_{j}\right)\\
	=\left(-i\xi_{1}\left|\eta\right|+L\right)\left(g\cdot e_{j}\right) & =\varrho_{j}\left(\eta\right)\left(g\cdot e_{j}\right).
\end{align*}
In this way, the original equation (\ref{eq:eigenprob}) is reduced to the following simplified problem:
\begin{equation}
\left(-i\xi_{1}\left|\eta\right|+L\right)\psi_{j}\left(\left|\eta\right|\right)=\varrho_{j}\left(\left|\eta\right|\right)\psi_{j}\left(\left|\eta\right|\right),\label{eq:eigenprob1}
\end{equation}
with $\psi_{j}\left(\left|\eta\right|\right)=\left(g\cdot e_{j}\right)\left(\eta\right)$.
Note that the dependence on $\eta$ is only through $\left|\eta\right|$.

	Next, it is natural to ask how the eigen-pairs $\{\varrho_i(|\eta|), \psi_{i}(|\eta|)\}, 0\leq i\leq 4$ for the reduced eigenvalue problem \eqref{eq:eigenprob1} depend on one dimensional parameter $|\eta|$. In fact, we are going to show that they are smooth in $|\eta|$ when $-2\leq\ga\leq 1$ and analytic in $|\eta|$ when $-1\leq\ga\leq 1$.
\begin{lemma}\label{smooth}
	For $|\eta|<\delta_{0}$, The eigenvalues $\varrho_i(|\eta|)$ and corresponding eigenfunctions $\psi_i(|\eta|)$, $0\leq i\leq 4$, are smooth in $|\eta|$ for $-2\leq \ga\leq 1$. Moreover, those are analytic in $|\eta|$ for $-1\leq \ga \leq 1$.
\end{lemma}
\begin{proof} The smoothness property of $\{\varrho_i(|\eta|), \psi_{i}(|\eta|)\}$ can be found in \cite{YangYu}. We only need to check that $\{\varrho_i(|\eta|), \psi_{i}(|\eta|)\}$ is analytic for $\gamma\geq -1$, i.e., the perturbation $i \xi g$ (in fact $i\xi_1 g$) is $L-$bounded:
	$$
	|\xi g|^{2}_{L^{2}_{\xi}}\leq C_{1} |Lg|^{2}_{L^{2}_{\xi}}+C_{2}|g|^{2}_{L^{2}_{\xi}}\,,
	$$
	for some $C_{1}, C_{2}>0$. Then the Kato-Rellich theorem \cite{[Kato]} guarantees this lemma. In order to prove this, let us calculate $\big<\Lambda g,\Lambda g\big>_{\xi}$ first. For simplicity of notation, let
	$$
	\psi(\xi)=\frac{1}{4}(\xi,\sigma\xi )+\varpi\chi_{R}-\frac{1}{2}\nabla_{\xi}\cdot\big[\sigma\xi\big]\,,
	$$
	then
	\begin{align*}
		\big<\Lambda g,\Lambda  g\big>_{\xi}&=\big|\nabla_{\xi}\cdot\big[\sigma\nabla_{\xi}g\big]\big|^{2}_{L^{2}_{\xi}}+|\psi(\xi)g|^{2}_{L^{2}_{\xi}}
		\\
		&\quad +2\big<\psi(\xi),\left( \nabla_{\xi}g, \sigma\nabla_{\xi}g\right)\big>_{\xi}+2\big<g ,\left(\nabla_{\xi}\psi(\xi), \sigma\nabla_{\xi}g\right)\big>_{\xi}\,.
	\end{align*}
Note that the third term is of good sign since $\sigma$ is a positive-definite matrix and $\psi(\xi)$ positive, we only need to control the last term. Integration by parts and Lemma \ref{prop1} give rise to
\[
\begin{aligned}
\quad 2\big<g ,\left(\nabla_{\xi}\psi(\xi), \sigma\nabla_{\xi}g\right)\big>_{\xi} & = \int \sigma^{ij}\partial_{i}\psi \partial_{j}(g^2)d\xi\\
&=-\int \nabla_{\xi}\cdot[\sigma\nabla_{\xi}\psi ]g^2 d\xi\\
&=-\int \nabla_{\xi}\cdot[\lambda_1(\xi)\nabla_{\xi}\psi ]g^2 d\xi\,,
\end{aligned}	
\]
this implies
\[
|2\big<g ,\left(\nabla_{\xi}\psi(\xi), \sigma\nabla_{\xi}g\right)\big>_{\xi}| \lesssim \left|\left<\xi\right>^{\ga}g\right|^{2}_{L^{2}_{\xi}}\,.
\]
It can be dominated by $|\psi(\xi)g|^{2}_{L^{2}_{\xi}}$ whenever $\varpi$ and $R$ are suitably large. This means
\begin{align*}
\big<\Lambda g,\Lambda g\big>_{\xi}\geq\big|\nabla_{\xi}\cdot\big[\sigma\nabla_{\xi}g\big]\big|^{2}_{L^{2}_{\xi}}+|\psi(\xi)g|^{2}_{L^{2}_{\xi}}\gtrsim
|\left<\xi\right>^{\ga+2}g|^{2}_{L^{2}_{\xi}}\,.
\end{align*}

Hence if $-1\leq\ga\leq1$,
	\begin{align*}
		|\xi g|^{2}_{L^{2}_{\xi}}\leq C \big<\Lambda g,\Lambda g\big>_{\xi}&=\big<L g-Kg,Lg-Kg\big>_{\xi}\\
		&\leq C_{1} |Lg|^{2}_{L^{2}_{\xi}}+C_{2}|g|^{2}_{L^{2}_{\xi}}\,.
	\end{align*}
	This completes the proof of the lemma.
\end{proof}

	Our goal is the original problem \eqref{eq:eigenprob}. Smooth (analytic) dependence on $|\eta|$ for the reduced problem \eqref{eq:eigenprob1} does not necessarily imply corresponding dependence on  $\eta$  for the original problem since magnitude function $|\cdot|$ is not smooth. This leads us to investigate more parity and conjugate properties of eigen-pairs.

Take the complex conjugate of \eqref{eq:eigenprob1} to have
\[
\left(-i\xi_{1}\left(-\left|\eta\right|\right)+L\right)\overline{\psi_{j}\left(\left|\eta\right|\right)}=\overline{\varrho_{j}\left(\left|\eta\right|\right)}\,\overline{\psi_{j}\left(\left|\eta\right|\right)}.
\]
Therefore the eigen-pair set $\left\{ \left(\overline{\varrho_{j}\left(\left|\eta\right|\right)},\overline{\psi_{j}\left(\left|\eta\right|\right)}\right)\right\} _{j=0}^{4}$
coincides with $\left\{ \left(\varrho_{j}\left(-\left|\eta\right|\right),\psi_{j}\left(-\left|\eta\right|\right)\right)\right\} _{j=0}^{4}$
. By checking their asymptotic expansions for $\left|\eta\right|\ll1$,
we conclude
\begin{align}
	\overline{\varrho_{j}\left(\left|\eta\right|\right)} & =\varrho_{j}\left(-\left|\eta\right|\right),\quad j=0,\dots,4,\label{eq:sym_eval_1}\\
	\overline{\psi_{j}\left(\left|\eta\right|\right)} & =\psi_{j}\left(-\left|\eta\right|\right),\quad j=0,\dots,4.\label{eq:sym_evec_1}
\end{align}

Define a map $\mathcal{R}:\left(\xi_{1},\xi_{2},\xi_{3}\right)\longmapsto\left(-\xi_{1},\xi_{2},\xi_{3}\right)$,
obviously $\mathcal{R}\in O\left(3\right)$ and $\mathcal{R}^{-1}=\mathcal{R}$.
Applying $\mathcal{R}$ to (\ref{eq:eigenprob1}),
\[
\left(-i\xi_{1}\left(-\left|\eta\right|\right)+L\right)\left(\mathcal{R}\cdot\psi_{j}\right)=\varrho_{j}\left(\left|\eta\right|\right)\left(\mathcal{R}\cdot\psi_{j}\right),
\]
which implies that two sets $\left\{ \left(\varrho_{j}\left(-\left|\eta\right|\right),\psi_{j}\left(-\left|\eta\right|\right)\right)\right\} _{j=0}^{4}$
and $\left\{ \left(\varrho_{j}\left(\left|\eta\right|\right),\left(\mathcal{R}\cdot\psi_{j}\right)\left(\left|\eta\right|\right)\right)\right\} _{j=0}^{4}$
are identical. Again we conclude
\begin{align}
	\varrho_{0}\left(-\left|\eta\right|\right) & =\varrho_{1}\left(\left|\eta\right|\right),\label{eq:sym_eval_2}\\
	\varrho_{j}\left(-\left|\eta\right|\right) & =\varrho_{j}\left(\left|\eta\right|\right),\mbox{ for }j=2,3,4,\label{eq:sym_eval_3}
\end{align}
and
\begin{align}
	\left(\mathcal{R}\cdot\psi_{0}\right)\left(-\left|\eta\right|\right) & =\psi_{1}\left(\left|\eta\right|\right),\label{eq:sym_evec_2}\\
	\left(\mathcal{R}\cdot\psi_{j}\right)\left(-\left|\eta\right|\right) & =\psi_{j}\left(\left|\eta\right|\right),\mbox{ for }j=2,3,4.\label{eq:sym_evec_3}
\end{align}

	Let us recall that the eigen-problem is an infinite dimensional problem. The eigenvalues $\varrho_{j}$'s are functions depending on $|\eta|$ only, while the eigenfunctions $\psi_{j}$ are functions of both $|\eta|$ and $\xi$. The resolution of this infinite dimensional eigen-problem is by reducing it to finite dimensional one, see \cite{[Degond],[LiuYu1],YangYu}. It is convenient to express the above symmetry properties in finite dimensional setting and recover the infinite dimensional one later.

Let $\mathrm{P}_0$ be the macroscopic projection. Then we have (see Lemma 7.7 \cite{[LiuYu1]} or Theorem 3.1 \cite{YangYu})
\[
\begin{aligned}
	\mathrm{P}_{0}\psi_{j}\left(\left|\eta\right|\right) & =\sum_{l=0}^{2}\beta_{jl}\left(\left|\eta\right|\right)E_{l}^{1},\:j=0,1,2,\\
	\mathrm{P}_{0}\psi_{3}\left(\left|\eta\right|\right) & =\beta_{33}\left(\left|\eta\right|\right)E_{3}^{1},\\
	\mathrm{P}_{0}\psi_{4}\left(\left|\eta\right|\right) & =\beta_{44}\left(\left|\eta\right|\right)E_{4}^{1}.
\end{aligned}
\]
Here $\left\{ E_{j}^{1}\right\} _{j=0}^{4}$ are normalized eigenvectors
for $\mathrm{P}_{0}\xi_{1}\mathrm{P}_{0}$, and they form an orthonormal
basis of $\mathrm{Ker}\,L$,
\begin{align*}
	E_{0}^{1} & =\sqrt{\frac{3}{10}}\chi_{0}+\sqrt{\frac{1}{2}}\chi_{1}+\sqrt{\frac{1}{5}}\chi_{4},\\
	E_{1}^{1} & =\sqrt{\frac{3}{10}}\chi_{0}-\sqrt{\frac{1}{2}}\chi_{1}+\sqrt{\frac{1}{5}}\chi_{4},\\
	E_{2}^{1} & =-\sqrt{\frac{2}{5}}\chi_{0}+\sqrt{\frac{3}{5}}\chi_{4},\\
	E_{3}^{1} & =\chi_{2},\\
	E_{4}^{1} & =\chi_{3}.
\end{align*}
In fact $E_{j}^{1}=g\cdot E_{j}$, where $E_{j}$ are defined in Lemma \ref{pr12}, $g\in O\left(3\right)$ and
sends $\frac{\eta}{\left|\eta\right|}$ to $\left(1,0,0\right)^{T}$.

The above symmetry properties \eqref{eq:sym_evec_1}, \eqref{eq:sym_evec_2}, and \eqref{eq:sym_evec_3} impose some restrictions on coefficients
$\beta_{ij}\left(\left|\eta\right|\right).$ We note $\mathcal{R}E_{0}^{1}=E_{1}^{1}$,
$\mathcal{R}E_{j}^{1}=E_{j}^{1}$, $j=2,3,4$. Then by (\ref{eq:sym_evec_2})
and Lemma \ref{lem:commute}
\[
\left(\mathcal{R}\cdot\mathrm{P}_{0}\psi_{0}\right)\left(-\left|\eta\right|\right)=\mathrm{P}_{0}\mathcal{R}\cdot\psi_{0}\left(-\left|\eta\right|\right)=\mathrm{P}_{0}\psi_{1}\left(\left|\eta\right|\right).
\]
The LHS is
\[
\beta_{00}\left(-\left|\eta\right|\right)E_{1}^{1}+\beta_{01}\left(-\left|\eta\right|\right)E_{0}^{1}+\beta_{02}\left(-\left|\eta\right|\right)E_{2}^{1},
\]
and the RHS is
\[
\beta_{10}\left(\left|\eta\right|\right)E_{0}^{1}+\beta_{11}\left(\left|\eta\right|\right)E_{1}^{1}+\beta_{12}\left(\left|\eta\right|\right)E_{2}^{1}.
\]
By comparing the coefficients we obtain
\[
\beta_{10}\left(\left|\eta\right|\right)=\beta_{01}\left(-\left|\eta\right|\right),\quad\beta_{11}\left(\left|\eta\right|\right)=\beta_{00}\left(-\left|\eta\right|\right),\quad\beta_{12}\left(\left|\eta\right|\right)=\beta_{02}\left(-\left|\eta\right|\right).
\]
Similarly, we find
\begin{align*}
	\beta_{21}\left(\left|\eta\right|\right) & =\beta_{20}\left(-\left|\eta\right|\right),\quad\beta_{22}\left(\left|\eta\right|\right)=\beta_{22}\left(-\left|\eta\right|\right),\\
	\beta_{33}\left(\left|\eta\right|\right) & =\beta_{33}\left(-\left|\eta\right|\right),\quad\beta_{44}\left(\left|\eta\right|\right)=\beta_{44}\left(-\left|\eta\right|\right).
\end{align*}
Moreover, (\ref{eq:sym_evec_1}) requires that
\[
\overline{\beta_{jk}\left(\left|\eta\right|\right)}=\beta_{jk}\left(-\left|\eta\right|\right).
\]
Given a function $\phi\left(s\right)$, $s\in\mathbb{R}$, there is
a natural way to decompose $\phi\left(s\right)$ into even and odd
parts,
\begin{align*}
	\phi\left(s\right) & =\frac{\phi\left(s\right)+\phi\left(-s\right)}{2}+\frac{\phi\left(s\right)-\phi\left(-s\right)}{2}\\
	& \equiv\phi^{e}\left(s\right)+\phi^{o}\left(s\right).
\end{align*}
Then it is easy to see
\[
\phi^{e}\left(s\right)=\mathcal{A}\left(s^{2}\right),\quad\frac{\phi^{o}\left(s\right)}{s}=\mathcal{B}\left(s^{2}\right)
\]
for some functions $\mathcal{A}$ and $\mathcal{B}$. Moreover $\mathcal{A}$
and $\mathcal{B}$ are smooth or analytic functions provided $\phi$
is smooth or analytic respectively. Indeed, this is an old result due
to Whitney:
\begin{proposition}[\cite{Whitney}]
	(1) An even function $f(x)$ may be written as $g(x^{2})$. If
	$f$ is analytic, of class $C^{\infty}$ or of class $C^{2s}$, $g$
	may be made analytic, of class $C^{\infty}$ or of class $C^{s}$,
	respectively.
	
	(2) An odd function $f(x)$ may be written as $xg(x^2).$If
	$f$ is analytic, of class $C^{\infty}$ or of class $C^{2s+1}$,
	$g$ may be made analytic, of class $C^{\infty}$ or of class $C^{s}$,
	respectively.
\end{proposition}
Based on the even-odd decomposition and symmetry properties, we can
express $\mathrm{P}_{0}\psi_{j}\left(\left|\eta\right|\right)$ very
explicitly. Moreover, the eigenfunctions can be reconstructed from
their macroscopic projection (see Liu-Yu \cite{[LiuYu1]}),
\begin{equation}
\psi_{j}\left(\left|\eta\right|\right)=\mathscr{L}_{j}^{1}\mathrm{P}_{0}\psi_{j}\left(\left|\eta\right|\right),\label{eq:eigenprob1-sol}
\end{equation}
where
\[
\mathscr{L}_{j}^{1}=1+\left(L-i\left|\eta\right|\mathrm{P}_{1}\xi_{1}-\varrho_{j}\left(\left|\eta\right|\right)\right)^{-1}\left(i\left|\eta\right|\mathrm{P}_{1}\xi_{1}\right).
\]

	So for, we have obtained detailed dependence on $|\eta|$ for simplified eigen-problem \eqref{eq:eigenprob1}. Now we are in a position to reconstruct the eigenfunctions $e_{j}$'s and their dependencies on $\eta$ associated
with original eigen-problem (\ref{eq:eigenprob}).
Since $g\cdot e_{j}=\psi_{j}$
and $g$ commutes with $\mathrm{P}_{1}$ and $L$, we apply $g^{-1}$
to obtain
\begin{align*}
	e_{j} & =g^{-1}\cdot\psi_{j}=g^{-1}\cdot\mathscr{L}_{j}^{1}\mathrm{P}_{0}\psi_{j}\\
	& =\left[1+\left(L-\mathrm{P}_{1}i\xi\cdot\eta-\varrho_{j}\left(\left|\eta\right|\right)\right)^{-1}\left(\mathrm{P}_{1}i\xi\cdot\eta\right)\right]\mathrm{P}_{0}g^{-1}\cdot\psi_{j}\\
	& \equiv\mathscr{L}_{j}\mathrm{P}_{0}e_{j}.
\end{align*}
Here we only calculate $e_{2}(\eta)$ as an illustration.
\begin{align*}
	\mathrm{P}_{0}\psi_{2}\left(\left|\eta\right|\right) & =\beta_{20}\left(\left|\eta\right|\right)E_{0}^{1}+\beta_{21}\left(\left|\eta\right|\right)E_{1}^{1}+\beta_{22}\left(\left|\eta\right|\right)E_{2}^{1}\\
	& =\beta_{20}\left(\left|\eta\right|\right)E_{0}^{1}+\beta_{20}\left(-\left|\eta\right|\right)E_{1}^{1}+\beta_{22}\left(\left|\eta\right|\right)E_{2}^{1}\\
	& =\left(2\sqrt{\frac{3}{10}}\beta_{20}^{e}\left(\left|\eta\right|\right)-\sqrt{\frac{2}{5}}\beta_{22}^{e}\left(\left|\eta\right|\right)\right)\chi_{0}+2\sqrt{\frac{1}{2}}\beta_{20}^{o}\left(\left|\eta\right|\right)\chi_{1}\\
	& \quad+\left(2\sqrt{\frac{1}{5}}\beta_{20}^{e}\left(\left|\eta\right|\right)+\sqrt{\frac{3}{5}}\beta_{22}^{e}\left(\left|\eta\right|\right)\right)\chi_{4}\\
	& \equiv\mathtt{a}_{2,1}^{0}(\left|\eta\right|^{2})\chi_{0}+i\left|\eta\right|\mathtt{a}_{2,2}^{1}(\left|\eta\right|^{2})\chi_{1}+\mathtt{a}_{2,1}^{4}(\left|\eta\right|^{2})\chi_{4},
\end{align*}
where $\mathtt{a}_{2,1}^{0}$, $\mathtt{a}_{2,2}^{1}$ and $\mathtt{a}_{2,1}^{4}$
are smooth or analytic function when $-2\leq\gamma<-1$ or $-1\leq\gamma\leq1$
respectively. Applying $g^{-1}$ and noting that
\[
\begin{aligned}
	g^{-1}\cdot\chi_{0} & =\chi_{0},\quad g^{-1}\cdot\chi_{4}=\chi_{4},\\
	g^{-1}\cdot\chi_{1} & =g^{-1}\cdot\mu^{1/2}\xi_{1} =\mu^{1/2}\left(g\xi\right)_{1}=\mu^{1/2}\xi\cdot\frac{\eta}{\left|\eta\right|} =\sum_{j=1}^{3}\frac{\eta_{j}}{\left|\eta\right|}\chi_{j},
\end{aligned}
\]
we obtain
\begin{align*}
	\mathrm{P}_{0}e_{2}\left(\eta\right) & =g^{-1}\cdot\mathrm{P}_{0}\psi_{2}\left(\left|\eta\right|\right)\\
	& =\mathtt{a}_{2,1}^{0}(\left|\eta\right|^{2})\chi_{0}+i\mathtt{a}_{2,2}^{1}(\left|\eta\right|^{2})\sum_{j=1}^{3}\eta_{j}\chi_{j}+\mathtt{a}_{2,1}^{4}(\left|\eta\right|^{2})\chi_{4},
\end{align*}
therefore
\[
e_{2}\left(\eta\right)=\mathscr{L}_{2}\left[\mathtt{a}_{2,1}^{0}(\left|\eta\right|^{2})\chi_{0}+i\mathtt{a}_{2,2}^{1}(\left|\eta\right|^{2})\sum_{j=1}^{3}\eta_{j}\chi_{j}+\mathtt{a}_{2,1}^{4}(\left|\eta\right|^{2})\chi_{4}\right].
\]
The other $e_{j}\left(\eta\right)$'s can be computed in a similar
manner and we omit the details for brevity. Therefore we have finished the proof.


\section{Smoothing effect (Proof of Lemma \ref{improve})}\label{Hypo}

In this section, we will prove Lemma \ref{improve}, which is the key lemma in this paper.

Let $f$ be the solution to the linearized Landau equation \eqref{in.1.c},
\begin{equation}
\label{prob}
\begin{cases}
\partial_t f +\xi\cdot\nabla_{x} f = \Lambda f + K f, \\
f(0,x,\xi)=f_0(x,\xi),
\end{cases}
\end{equation}
where operators $\Lambda$ and $K$ are defined in Lemma \ref{prop1}. Let us recall that
\[
\Lambda f= \nabla_\xi \cdot \left[\sigma \nabla_\xi f\right]-\psi(\xi)f,
\]where
\[
\begin{aligned}
\psi(\xi) &\equiv  \varpi \chi_{R}(|\xi|) +\frac{1}{4}(\xi,\sigma\xi)-\frac{1}{2} \nabla_\xi \cdot \left[\sigma\xi\right]\\
&= \varpi \chi_{R}(|\xi|)+ \left\{ \frac{3}{4}\lambda_1(\xi) \left|\xi\right|^2 -\lambda_2(\xi) -\frac{1}{2}\lambda_1(\xi)\right\}.
\end{aligned}
\]
And
\[
Kf=\widetilde{K}f+\varpi\chi_R(|\xi|)f.
\]


Let $h$ be the solution to the equation without operator $K$, which captures the initial singularity of the original solution $f$:
\begin{equation}
\label{prob_noweight}
\left\{
\begin{aligned}
&\partial_t h=\mathcal{L} h,\quad\mbox{where }\mathcal{L}h=-\xi\cdot \nabla_x h + \nabla_{\xi}\cdot[\sigma\nabla_{\xi}h]-\psi(\xi)h,\\
&h(0,x,\xi)=h_0(x,\xi).
\end{aligned}
\right.
\end{equation}
From now on, the notations $\pa_{k}$ and $\pa_{ij}$ mean the first and second derivatives in the $\xi$-variable, i.e., $\pa_{k}=\pa_{\xi_{k}}$ and $\pa_{ij}=\pa^{2}_{\xi_{i}\xi_{j}}$. For simplicity, in Lemma \ref{hypodiss} and Lemma \ref{regul}, we denote the integral
$$
\int_{\R^{6}}\cdot\cdot\cdot d\xi dx \equiv \int \cdot\cdot\cdot\,.
$$
We choose $m_{0}(\xi)=\langle\xi\rangle^k$, where $k\in\mathbb{R}$. Note that $|\partial_\xi^s m_{0}|\lesssim \langle \xi \rangle^{k-|s|}.$

\subsection{Hypodissipativity}We prove the hypodissipativity of operator $\mathcal{L}$ firstly.
\begin{lemma}[Hypodissipativity]\label{hypodiss}
	Let $h$ be the solution to \eqref{prob_noweight}. For $l\in \N$, we can choose $\varpi>0$, $R>0$ large enough, $1/D$ and $\alpha$ suitably small such that
	\begin{equation}
	\left\|e^{t\mathcal{L}}h_0\right\|_{H^{l}_{x}L^{2}_{\xi}(\varrho \,m_0 )} \lesssim e^{-\lambda t}
	\left\|h_{0} \right\|_{H^{l}_{x}L^{2}_{\xi}(\varrho\,m_0)}
	\end{equation}
	for some positive constant $\lambda$.
\end{lemma}

\begin{proof}
	For proof of the case $\varrho(x,\xi)=1$, we refer the reader to \cite{[CTK]}. Let $h(t)=e^{t\mathcal{L}}h_0$. Since $x$-derivative commutes with operator $\mathcal{L}$, it suffices to complete the $L^2_xL^2_\xi(\varrho\, m_0)$ estimate.
	\[
	\frac{d}{dt} \frac{1}{2} \int h^2 \varrho\, m_0
 =\int h\mathcal{L} h \varrho\, m_0    = -\int h (\xi\cdot \nabla_x h) \varrho\, m_0 + \int h\nabla_{\xi}\cdot[\sigma\nabla_{\xi}h]\varrho\, m_0 -\int \psi(\xi) h^2 \varrho\, m_0.
	\]
	Using integration by parts, the first two terms of RHS can be rewritten as follows,
	\[
	 -\int h (\xi\cdot \nabla_x h) \varrho\, m_0=\frac{1}{2} \int \left(\xi\cdot\nabla_x(\varrho\, m_0)\right)h^2,
	\]
	and
	\[
	\begin{aligned}
	&\quad\int h\nabla_{\xi}\cdot[\sigma\nabla_{\xi}h]\varrho\, m_0\\
	& =   -\int (\sigma \nabla_{\xi}h,\nabla_\xi h)\varrho\, m_0  - \int h (\sigma\nabla_{\xi}h,\nabla_{\xi}(\varrho\, m_0))\\
	& =	  -\int (\sigma \nabla_{\xi}h,\nabla_\xi h)\varrho\, m_0 + \frac{1}{2}\int (\partial_i\sigma^{ij})\partial_j(\varrho\, m_0) h^2 + \frac{1}{2} \int \sigma^{ij} \partial_{ij}(\varrho\, m_0) h^2.
	\end{aligned}
	\]
	Combing the above equations together, we arrive at
	\[
	\begin{aligned}
	\frac{1}{2}\frac{d}{dt}\int h^2 \varrho\, m_0 & =  -\int (\sigma \nabla_{\xi}h,\nabla_\xi h)\varrho\, m_0 \\
	&\quad - \int h^2 \varrho\, m_0 \underbrace{\left\{\psi(\xi) -\frac{1}{2}\frac{\xi\cdot \nabla_x(\varrho\, m_0)}{\varrho\, m_0}-\frac{1}{2} \frac{(\partial_i\sigma^{ij})\partial_j(\varrho\, m_0)}{\varrho\, m_0} -\frac{1}{2}\frac{ \sigma^{ij} \partial_{ij}(\varrho\, m_0)}{\varrho\, m_0}	 \right\}}_{=:\Psi(x,\xi)}.
	\end{aligned}
	\]
	We discuss for different ranges of $\gamma$:
	\newline \textbf{Case 1}. $\gamma\in[-1,1]$: $\varrho=\exp(\frac{\left<x\right>}{D})$, and
	\[
	\nabla_x(\varrho\, m_0)=\frac{1}{D}\frac{x}{\left<x\right>}\varrho\, m_0,\quad (\partial_i \sigma^{ij})\partial_j(\varrho\, m_0)=-\lambda_1(\xi)\xi_j \varrho \partial_j m_0,\quad \partial_{ij}(\varrho\, m_0)=\varrho \partial_{ij}m_0,
	\]
	thus
	\[
	\Psi(x,\xi)\geq c_0\left<\xi\right>^{\gamma+2}-C\left(\frac{\left<\xi\right>}{D}+\left<\xi\right>^{\gamma}\right)\gtrsim \left<\xi\right>^{\gamma+2}
	\]
	provided $D,\varpi,R$ are chosen large. This then follows
	\[
	\frac{d}{dt}\int h^2 \varrho\, m_0 \lesssim -\left\|h \right\|^2_{L^2_\sigma(\varrho\, m_0)}\leq - \left\| h\right\|^2_{L^2(\varrho\, m_0)}.
	\]
	\newline \textbf{Case 2}. $\gamma\in[-2,-1)$: $\varrho=\exp(\alpha \vartheta(0,x,\xi))$ with $\vartheta(t,x,\xi)$ defined in \eqref{weight-functions}. Direct computations show
	\[
	\begin{aligned}
	\Psi(x,\xi) & = \psi(\xi)-\frac{1}{2}\alpha \xi\cdot \nabla_x \vartheta+\frac{1}{2}\lambda_1(\xi)\xi_j\left(\alpha\partial_j \vartheta+\frac{\partial_j m_0}{m_0}\right)\\
			& \quad - \frac{1}{2}\sigma^{ij}\left(\alpha^2 \partial_i \vartheta \partial_j \vartheta+ \alpha \partial_{ij}\vartheta+\alpha \frac{\partial_i \vartheta\partial_j m_0 + \partial_j \vartheta\partial_i m_0}{m_0}+\frac{\partial_{ij}m_0}{m_0}\right)\,.
	\end{aligned}
	\]
	Noticing that $\varrho,m_0$ are functions of $\xi$ through $\left<\xi\right>$, we get $\partial_j \varrho= \frac{\partial\varrho}{\partial\left<\xi\right>}\frac{\xi_j}{\left<\xi\right>}$, and similarly for $m_0$. In view of that $\xi$ is an eigenvector of  matrix $\sigma$ associated with eigenvalue $\lambda_1(\xi)$ and the following estimates
	\[
	\left|\frac{\partial^s\vartheta}{\partial\left<\xi\right>^s}\right|\lesssim \left<\xi\right>^{2-s}~ \mbox{for}~s\in\mathbb{N}^+,\quad \left|\frac{1}{m_0}\frac{\partial^s m_0}{\partial\left<s\right>^s} \right|\lesssim \left<\xi\right>^{-s}~\mbox{for}~s\in\mathbb{N},
	\]
	we find
	\[
	\Psi(x,\xi)\geq c_0\left<\xi\right>^{\gamma+2}-C\left(\alpha\delta\left<\xi\right>^{\gamma+2}+ \alpha \left<\xi\right>^{\gamma+2} + \left<\xi\right>^{\gamma}+\alpha^2\left<\xi\right>^{\gamma+2} \right)\gtrsim\left<\xi\right>^{\gamma+2}
	\]
	if choosing $\alpha,\delta$ small and $\varpi,R$ large. Again we have
	\begin{equation}
	\label{eq:hypo-1}
	\frac{d}{dt}\int h^2 \varrho\, m_0 \lesssim -\left\|h \right\|^2_{L^2_\sigma(\varrho\, m_0)}\leq - \left\| h\right\|^2_{L^2(\varrho\, m_0)}.
	\end{equation}
	The proof is therefore complete.
\end{proof}

\subsection{Regularization} In this subsection, we will show the regularization property of the semigroup $e^{t\mathcal{L}}$ for small time.

\begin{lemma}\label{regul}
Let $h$ be the solution to \eqref{prob_noweight} and $m_0=\left<\xi\right>^{k}$ for some $k\in \mathbb{N}$. Define
\[
m_n\equiv \begin{cases}
m_0 & \ga\in\left[0,1\right],\\
\left<\xi\right>^{n|\ga|} m_0 & \ga \in \left[-2,0\right).
\end{cases}
\]
Then for $0<t\leq 1$, we have
\[
\left\|\nabla_{x} e^{t\mathcal{L}} h_0\right\|_{L^{2}(\varrho m_{0})}\lesssim t^{-3/2}\left\|h_{0} \right\|_{L^{2}(\varrho m_{1})}\,,\quad
\left\|\nabla_{\xi} e^{t\mathcal{L}} h_0 \right\|_{L^{2}(\varrho m_{0})}\lesssim t^{-1/2}\left\|h_{0} \right\|_{L^{2}(\varrho m_{1})}\,.
\]
\end{lemma}

\begin{proof}
Let $h(t)=e^{t\mathcal{L}}h_0$. We follow the technique introduced by H\'erau \cite{[Herau]} (see also very recent results \cite{[CTK]} for the Landau equation and \cite{[MisMou]} for the Fokker-Planck equation) to define the functional
\begin{equation}
\begin{aligned}
\mathcal{F}(t,h)&\equiv \int h^2 \varrho\, m_1  +\alpha_1 t \int \left|\nabla_\xi h\right|^2 \varrho\, m_0 + \alpha_2 t^2 \int \nabla_x h\cdot\nabla_{\xi} h\,\varrho \,m_0 \\
&\quad + \alpha_3 t^3 \int \left|\nabla_x h\right|^2 \varrho\, m_0,
\end{aligned}
\end{equation}
where $\alpha_i>0$, $i=1,2,3$ will be chosen later.

Differentiating in $t$, we have
\begin{equation}\label{eq_l6.1}
\begin{aligned}
\frac{d}{d t} \mathcal{F}(t,h(t)) & =   \frac{d}{dt }\int h^2 \varrho\, m_1 + \alpha_1  \int \left|\nabla_\xi h\right|^2 \varrho\,m_0+ \alpha_1 t \frac{d}{dt}\int \left|\nabla_\xi h\right|^2 \varrho\,m_0\\
& \quad + 2\alpha_2 t \int \nabla_x h\,\cdot\nabla_{\xi} h\,\varrho\,m_0 + \alpha_2 t^2 \frac{d}{dt}  \int \nabla_x h\cdot\nabla_{\xi} h\,\varrho\,m_0\\
& \quad +3\alpha_3 t^2 \int \left|\nabla_x h\right|^2 \varrho \,m_0 + \alpha_3 t^3 \frac{d}{dt} \int \left|\nabla_x h\right|^2 \varrho \, m_0.
\end{aligned}
\end{equation}
The first and the last terms have been calculated in previous Lemma \ref{hypodiss}. We only need to compute the third $\xi$-derivative term and the fifth mixed term.

\noindent \underline{\it Step 1: $\xi$-derivative estimate.} The energy estimate gives
\[
\frac{1}{2}\frac{d}{dt} \int \left|\nabla_\xi h\right|^2 \varrho\,m_{0}=\int (\nabla_\xi h, \mathcal{L}\nabla_\xi h )\varrho\, m_{0} + \int (\nabla_\xi h, \left[\nabla_\xi, \mathcal{L}\right]h )\varrho\, m_{0} \,.
\]
For the first term, applying the previous $L^2$ estimate in Lemma \ref{hypodiss} results in
\[
\int (\nabla_\xi h, \mathcal{L}\nabla_\xi h)\varrho\,m_{0} =-\int \left(\sigma\nabla_\xi\partial_k h,\nabla_\xi\partial_k h\right)\varrho\, m_{0} - \int \Psi(x,\xi) \left|\nabla_\xi h\right|^2 \varrho\, m_{0} .
\]
We compute the commutator in the second term to get:
\begin{equation}\label{comm_x1}
\left[\nabla_\xi,\mathcal{L}\right]h =-\nabla_{x}h + \partial_i\left[(\partial_k\sigma^{ij}) \partial_j h \right] -(\nabla_{\xi}\psi)h.
\end{equation}
Substituting back and integration by parts give
\[
\begin{aligned}
\int (\nabla_{\xi}h, [\nabla_{\xi},\mathcal{L}]h)\varrho\,m_0 & =-\int (\nabla_{\xi}h,\nabla_x h)\varrho\,m_0 - \int (\partial_k \sigma^{ij}) (\partial_j h) (\partial_{ik}h)\varrho\, m_0 \\
&\quad - \int (\partial_k \sigma^{ij})(\partial_j h)(\partial_k h) \partial_i(\varrho\,m_0) + \frac{1}{2}\int h^2 \varrho\, m_0 \left(\Delta_{\xi}\psi + \frac{(\nabla_{\xi}\psi\,,\nabla_{\xi}(\varrho\, m_0))}{\varrho \, m_0}	\right).
\end{aligned}
\]

Gathering the above equations together, we obtain
\begin{equation}\label{eq_l5.1}
\begin{aligned}
&\quad \frac{1}{2}\frac{d}{dt} \int \left|\nabla_\xi h\right|^2 \varrho\, m_{0}\\
&=-\int \left(\sigma\nabla_\xi\partial_k h,\nabla_\xi\partial_k h\right)\varrho\,m_{0} - \int \Psi(x,\xi) \left|\nabla_\xi h\right|^2\varrho\, m_{0} \\
&\quad -\int (\nabla_{\xi}h,\nabla_x h)\varrho\,m_0 - \int (\partial_k \sigma^{ij}) (\partial_j h) (\partial_{ik}h)\varrho\, m_0- \int (\partial_k \sigma^{ij})(\partial_j h)(\partial_k h) \partial_i(\varrho\,m_0) \\
&\quad  +\frac{1}{2}\int h^2 \varrho\, m_0 \left(\Delta_{\xi}\psi + \frac{(\nabla_{\xi}\psi\,,\nabla_{\xi}(\varrho\, m_0))}{\varrho \, m_0}	\right) .
\end{aligned}
\end{equation}
The last two terms are easy to estimate:
\newline \textbf{Case 1}. $\gamma\in[-1,1]$: $\varrho=\exp(\frac{\left<x\right>}{D})$.
\begin{align*}
&\left|	 \int (\partial_k \sigma^{ij})(\partial_j h)(\partial_k h) \partial_i(\varrho\,m_0) 	\right|\lesssim \int \left<\xi\right>^{\gamma} \left|\nabla_{\xi}h\right|^2 \varrho\, m_0,\\
& \left|\int h^2 \varrho\, m_0 \left(\Delta_{\xi}\psi + \frac{(\nabla_{\xi}\psi\,,\nabla_{\xi}(\varrho\, m_0))}{\varrho \, m_0}	\right) \right|\lesssim \int \left<\xi\right>^{\gamma} h^2\varrho\,m_0.
\end{align*}
\newline \textbf{Case 2}. $\gamma\in[-2,-1)$: $\varrho=\exp(\alpha \vartheta(0,x,\xi))$ with $\vartheta(t,x,\xi)$ defined in \eqref{weight-functions}. Thanks to
\[
\left|\nabla_{\xi}\varrho\right|=\left|\alpha\varrho \nabla_{\xi}\vartheta(0,x,\xi)\right|\lesssim \alpha \varrho \left<\xi\right>,
\]
we have
\begin{align*}
&\left|	 \int (\partial_k \sigma^{ij})(\partial_j h)(\partial_k h) \partial_i(\varrho\,m_0) 	\right|\lesssim \alpha \int \left<\xi\right>^{\gamma+2} \left|\nabla_{\xi}h\right|^2 \varrho\, m_0,\\
&\left|\int h^2 \varrho\, m_0 \left(\Delta_{\xi}\psi + \frac{(\nabla_{\xi}\psi\,,\nabla_{\xi}(\varrho\, m_0))}{\varrho \, m_0}	\right) \right|\lesssim \int \left(\left<\xi\right>^{\gamma}+\alpha \left<\xi\right>^{\gamma+2}\right) h^2\varrho\,m_0\,.
\end{align*}

The estimate of the fourth term $\int (\partial_k \sigma^{ij}) (\partial_j h) (\partial_{ik}h)\varrho\, m_0 $ in \eqref{eq_l5.1} needs more efforts. We want to make use of the anisotropic property of $\sigma$ matrix and decompose the vector accordingly. To this end, we introduce the cut-off function $\chi(|\xi|)$ to decompose the integral domain and transfer the $\xi$-derivative on $\sigma$ to other terms when $\xi$ away from $0$.
\[
\begin{aligned}
&\quad\int (\partial_k \sigma^{ij}) (\partial_j h) (\partial_{ik}h)\varrho\, m_0\\
&=\int (\partial_k \sigma^{ij}) (\partial_j h) (\partial_{ik}h)\chi(|\xi|)\varrho\, m_0 +\int (\partial_k \sigma^{ij}) (\partial_j h) (\partial_{ik}h)(1-\chi(|\xi|))\varrho\, m_0\\
&=\int (\partial_k \sigma^{ij}) (\partial_j h) (\partial_{ik}h)\chi\,\varrho\, m_0 -\int \sigma^{ij} (\partial_{jk}h)(\partial_{ik}h)(1-\chi)\varrho\,m_0 \\
&\quad -\int \sigma^{ij} (\partial_{j}h)(\partial_{ik}h)\partial_{k}((1-\chi)\varrho\,m_0) - \int \sigma^{ij} (\partial_{j}h)(\partial_{ikk}h)(1-\chi)\varrho\,m_0 \\
&=: T_1 +T_2 +T_3 +T_4.
\end{aligned}
\]
For $T_1$, thanks to cut-off function $\chi(|\xi|)$, we have $\left<\xi\right>^{\gamma+1}\sim \left<\xi\right>^{\gamma}$ when $|\xi|\leq 2$, thus
\[
\begin{aligned}
\left|T_1\right|
& \lesssim \int \left<\xi\right>^{\gamma} \left|\nabla_{\xi}h\right|	\left|D^2_\xi h\right| \varrho\, m_0\\
& \lesssim \varepsilon \left\|\left<\xi \right>^{\frac{\gamma}{2}} \left|D^2_\xi h\right| \right\|_{L^2(\varrho\,m_0)}+C(\varepsilon) \left\|\left<\xi \right>^{\frac{\gamma}{2}} \left|\nabla_\xi h\right| \right\|_{L^2(\varrho\,m_0)}.
\end{aligned}
\]
For $T_{4}$, we decompose the vectors with respect to the eigenspace of the matrix $\sigma$, i.e.,  parallel to $\xi$ part and perpendicular to $\xi$ part and apply integration by parts,
\[
\begin{aligned}
T_4 & = -\int \Big\{ \lambda_1(\xi) \mathbb{P}(\xi)\nabla_{\xi}h\cdot \mathbb{P}(\xi) \nabla_{\xi}\partial_{kk}h	+	\lambda_2(\xi) (I_3-	 \mathbb{P}(\xi))\nabla_{\xi}h\cdot (I_3	-\mathbb{P}(\xi)) \nabla_{\xi}\partial_{kk}h \Big\}(1-\chi(|\xi|))\varrho\,m_0\\
	& = \int \Big\{( \partial_k \lambda_1 ) \mathbb{P}\nabla_{\xi}h\cdot \mathbb{P} \nabla_{\xi}\partial_{k}h	+	(\partial_k\lambda_2) (I_3-	 \mathbb{P})\nabla_{\xi}h\cdot (I_3	-\mathbb{P}) \nabla_{\xi}\partial_{k}h \Big\}(1-\chi)\varrho\,m_0 \\
	& \quad+ \int \left\{
	\begin{array}[c]{l}
	\lambda_1 \big[(\partial_{k}\mathbb{P})\nabla_{\xi}h\cdot \mathbb{P} \nabla_{\xi}\partial_{k}h	+ \mathbb{P}\nabla_{\xi}h\cdot (\partial_k \mathbb{P}) \nabla_{\xi}\partial_{k}h	\big]\\
		+	\lambda_2 \big[(\partial_{k}(I_3-\mathbb{P}))\nabla_{\xi}h\cdot (I_3- \mathbb{P}) \nabla_{\xi}\partial_{k}h	+ (I_3-\mathbb{P})\nabla_{\xi}h\cdot (\partial_k(I_3- \mathbb{P})) \nabla_{\xi}\partial_{k}h	\big]
	\end{array} \right\}(1-\chi)\varrho\,m_0\\
	&\quad +  \int \Big\{  \lambda_1 \mathbb{P}\nabla_{\xi}\partial_{k} h\cdot \mathbb{P} \nabla_{\xi}\partial_{k}h	+	\lambda_2 (I_3-	 \mathbb{P})\nabla_{\xi}\partial_{k}h\cdot (I_3	-\mathbb{P}) \nabla_{\xi}\partial_{k}h \Big\}(1-\chi)\varrho\,m_0\\
	&\quad +  \int \Big\{  \lambda_1 \mathbb{P}\nabla_{\xi} h\cdot \mathbb{P} \nabla_{\xi}\partial_{k}h	+	\lambda_2 (I_3-	 \mathbb{P})\nabla_{\xi}h\cdot (I_3	-\mathbb{P}) \nabla_{\xi}\partial_{k}h \Big\}\partial_{k}((1-\chi)\varrho\,m_0)\\
	&=: T_{41}+T_{42}+T_{43}+T_{44}.
\end{aligned}
\]
Noticing that
\[
T_2=-T_{43},\quad T_3=-T_{44},
\]
we get
\[
\int (\partial_k \sigma^{ij}) (\partial_j h) (\partial_{ik}h)\varrho\, m_0 =T_1+T_2+T_3+T_4=T_1+T_{41}+T_{42}.
\]
Using $\left|\nabla_{\xi}\lambda_1\right|\lesssim\left<\xi\right>^{\gamma-1}$, $\left|\nabla_{\xi}\lambda_2\right|\lesssim\left<\xi\right>^{\gamma+1}$,  $\left|\nabla_{\xi}\mathbb{P}\right|\lesssim \left<\xi\right>^{-1}$ when $|\xi|>1$ and Young's inequality, $T_{41}$ and $T_{42}$ are bounded by
\begin{align*}
\left|T_{41}\right|,\left|T_{42}\right|
&	\lesssim \int \left<\xi\right>^{\gamma-1} \left|\nabla_{\xi}h\right|\left|\mathbb{P}\nabla_{\xi}\partial_{k}h\right|\varrho\,m_0
	+  \int \left<\xi\right>^{\gamma-1} \left|\mathbb{P}\nabla_{\xi}h\right|\left|D^2_{\xi}h\right|\varrho\,m_0 \\
&\quad +  \int \left<\xi\right>^{\gamma+1} \left|\nabla_{\xi}h\right|\left|(I_3-\mathbb{P})\nabla_{\xi}\partial_{k}h\right|\varrho\,m_0
	+ \int \left<\xi\right>^{\gamma+1} \left|(I_3-\mathbb{P})\nabla_{\xi}h\right|\left|D^2_{\xi}h\right|\varrho\,m_0 \\
&\lesssim \varepsilon \left\|\left<\xi\right>^{\frac{\gamma}{2}} D^2_{\xi}h \right\|^2_{L^2(\varrho\,m_0)}		+		\varepsilon \left\|\left<\xi\right>^{\frac{\gamma+2}{2}} (I-\mathbb{P})\nabla_{\xi}\partial_k h \right\|^2_{L^2(\varrho\,m_0)}\\
&\quad + C(\varepsilon) \left\|\left<\xi\right>^{\frac{\gamma}{2}} \nabla_{\xi}h \right\|^2_{L^2(\varrho\,m_0)}		+		C(\varepsilon )\left\|\left<\xi\right>^{\frac{\gamma+2}{2}} (I-\mathbb{P})\nabla_{\xi} h \right\|^2_{L^2(\varrho\,m_0)}.\\
\end{align*}
Collecting above inequalities and choosing $\alpha$ and $\varepsilon$ suitably small, we have
\begin{equation}\label{eq_l6.6}
\begin{aligned}
&\quad\frac{d}{dt} \int \left|\nabla_\xi h\right|^2 \varrho\, m_{0}\\
& \leq -c_{0}\left\| \nabla_\xi h \right\|^2_{L^2_\sigma(\varrho\, m_{0})}+ C \left\|h \right\|^2_{L^2_\sigma(\varrho\,m_0)}+ C\int \left|\nabla_x h\right|\left|\nabla_{\xi} h\right|\varrho\, m_{0} \\
& \leq  -c_{0}\left\| \nabla_\xi h \right\|^2_{L^2_\sigma(\varrho\, m_{0})}+ C \left\|h \right\|^2_{L^2_\sigma(\varrho\,m_0)} + C\varepsilon_1 t \left\|\nabla_{x}h \right\|^2_{L^2(\varrho \,m_0)}+ C (\varepsilon_1 t)^{-1} \left\|\nabla_{\xi}h \right\|^2_{L^2(\varrho \,m_0)}.
\end{aligned}
\end{equation}
\noindent \underline{\it Step 2: mixed term estimate.}
\begin{equation}
\label{eq:reg-0}
\begin{aligned}
 \frac{d}{dt} \int (\nabla_{x} h\, ,\nabla_{\xi} h)\varrho \,m_{0} &= \int (\nabla_{x}\mathcal{L}h\,,\nabla_{\xi} h )\varrho\,m_{0}  + \int (\nabla_{x} h\,,\nabla_{\xi} \mathcal{L}h)\varrho\, m_{0} \\
&= \int \left\{ \left(\mathcal{L}\nabla_{x} h, \nabla_{\xi} h\right)+  \left(\nabla_{x} h, \mathcal{L}\nabla_{\xi} h\right)  \right\}\varrho\,m_{0} \\
&\quad 	+\int \left(\left[\nabla_{x},\mathcal{L}\right]h, \nabla_{\xi} h\right)\varrho\,m_{0}	 + \int  \left(\nabla_{x} h ,\left[\nabla_{\xi}, \mathcal{L}\right]h\right)\varrho\,m_{0}  \\
&\equiv I_1+I_2+I_3.
\end{aligned}
\end{equation}

Let us calculate $I_1$ firstly. Given two functions $f$  and $g$, we have
\[
\mathcal{L}(fg)=f \mathcal{L}g+ g\mathcal{L}f + \psi\,f\,g+2\left(\nabla_\xi f,\sigma \nabla_\xi g\right).
\]
In addition, use integration by parts to yield
\[
\begin{aligned}
\int \mathcal{L} f \varrho\, m_{0}  & = -\int (\xi\cdot \nabla_{x} f)\varrho \, m_{0}  + \int \nabla_\xi\cdot\left[\sigma\nabla_\xi f\right]\varrho\, m_{0}-\int \psi f \varrho\,m_{0}   \\
&= \int\left( \frac{\xi\cdot \nabla_x (\varrho m_0) + (\partial_i \sigma^{ij})\partial_j(\varrho m_0)+ \sigma^{ij}\partial_{ij}(\varrho m_0) }{\varrho m_0} 	-\psi(\xi)	\right) f\varrho\,m_0.
\end{aligned}
\]
Hence
\begin{equation}
\label{eq:reg-1}
\begin{aligned}
I_1 & = \int \left\{\mathcal{L} \left(\nabla_{x} h\, ,\nabla_{\xi} h\right)-\psi\left(\nabla_{x} h\, ,\nabla_{\xi} h\right)-2 \left(\nabla_\xi \pa_{k} h, ,\sigma\nabla_\xi \pa_{x_{k}} h\right) \right\}\varrho\,m_{0} \\
&= \int\left\{\frac{\xi\cdot \nabla_x (\varrho m_0) + (\partial_i \sigma^{ij})\partial_j(\varrho m_0)+ \sigma^{ij}\partial_{ij}(\varrho m_0) }{\varrho m_0} 	-2\psi(\xi)	\right\} \left(\nabla_{x} h\, ,\nabla_{\xi} h\right) \varrho \,m_{0} \\
&\quad - 2\int \left(\sigma\nabla_\xi \pa_{k} h,\nabla_\xi \pa_{x_{k}} h\right) \varrho\, m_{0} \\
&\lesssim \int \left<\xi\right>^{\gamma+2} \left|\nabla_x h \right| \left|\nabla_{\xi} h\right|\varrho\,m_0 + \int \left<\xi\right>^{\gamma} \left|\mathbb{P}\nabla_{\xi}\partial_{k}h\right|\left|\mathbb{P}\nabla_{\xi}\partial_{x_k}h\right|\varrho\, m_0\\
&\quad + \int \left<\xi\right>^{\gamma+2}\left|(I_3-\mathbb{P})\nabla_{\xi}\partial_{k}h\right|\left|(I_3-\mathbb{P})\nabla_{\xi}\partial_{x_k}h\right|\varrho\, m_0\,.
\end{aligned}
\end{equation}

For $I_2$, thanks to that $\nabla_x$ commutes with $\mathcal{L}$ operator, we immediately have
\begin{equation}
\label{eq:reg-2}
I_2=0.
\end{equation}

Now we deal with $I_3$. From \eqref{comm_x1}, we get
\begin{equation}
\label{eq:reg-3}
\begin{aligned}
I_3
& = \int  \left(\nabla_{x}h, [\nabla_{\xi},\mathcal{L}]h\right)\varrho\,m_0\\
& = -\int \left|\nabla_{x} h\right|^2 \varrho m_0 +\underbrace{\int (\partial_{x_{k}}h) \partial_{i}[(\partial_{k}\sigma^{ij})\partial_{j}h ]\varrho\,m_0}_{=:T_1} \underbrace{-\int h (\nabla_{\xi}\psi,\nabla_{x} h)\varrho\,m_0}_{=:T_2}.
\end{aligned}
\end{equation}
By integration by parts,
\[
T_1 = -\int (\partial_{k}\sigma^{ij})(\partial_{i}\partial_{x_{k}} h) (\partial_j h)\varrho\,m_0 - \int (\partial_{k}\sigma^{ij})(\partial_{x_{k}}h)(\partial_{j} h) \partial_{i}(\varrho\,m_0)=: T_{11}+T_{12}.
\]
Direct computations show for $\gamma\in [-1,1]$,
\[
\left|T_{12}\right|\lesssim \int \left< \xi \right>^{\gamma} \left|\nabla_{x}h\right|\left|\nabla_{\xi} h\right|\varrho m_0,
\]
and for $\gamma\in [-2,-1)$,
\[
\left|T_{12}\right|\lesssim \int \left[\alpha \left<\xi\right>^{\gamma+2}+\left< \xi \right>^{\gamma} \right]\left|\nabla_{x}h\right|\left|\nabla_{\xi} h\right|\varrho m_0,
\]
so they are both bounded by
\[
\left|T_{12}\right|\lesssim \int  \left<\xi\right>^{\gamma+2}\left|\nabla_{x}h\right|\left|\nabla_{\xi} h\right|\varrho m_0.
\]

We apply similar calculation as those in step 1 to $T_{11}$,  i.e., introduce cut-off function $\chi(|\xi|)$ and make use of the anisotropic property of $\sigma$ matrix:
\[
\begin{aligned}
T_{11}
& = -\int (\partial_{k}\sigma^{ij})(\partial_{i}\partial_{x_{k}} h) (\partial_j h)\chi(|\xi|)\varrho\,m_0	-\int (\partial_{k}\sigma^{ij})(\partial_{i}\partial_{x_{k}} h) (\partial_j h)(1-\chi(|\xi|))\varrho\,m_0\\
& =   -\int (\partial_{k}\sigma^{ij})(\partial_{i}\partial_{x_{k}} h) (\partial_j h)\chi(|\xi|)\varrho\,m_0 + \int \sigma^{ij} (\partial_{i}\partial_{x_{k}} \partial_{k} h) (\partial_j h)(1-\chi(|\xi|))\varrho\,m_0\\
&\quad +\int \sigma^{ij} (\partial_{i}\partial_{x_{k}} h) (\partial_j\partial_{k}  h)(1-\chi(|\xi|))\varrho\,m_0 + \int \sigma^{ij} (\partial_{i}\partial_{x_{k}}  h) (\partial_j h)\partial_{k}\left((1-\chi(|\xi|))\varrho\,m_0\right)\\
& =: T_{111}+T_{112}+T_{113}+T_{114}.
\end{aligned}
\]
For $T_{111}$, since $\left<\xi\right>^{\gamma+1}\sim\left<\xi\right>^{\gamma}$ for $|\xi|\leq 2$, we have
\[
T_{111} \lesssim \int \left<\xi\right>^\ga \left| \nabla_\xi h\right| \left|\nabla_{x}\nabla_\xi h\right|\varrho\,m_{0}.
\]
For $T_{112}$, we decompose $\nabla_{\xi}\partial_{k}\partial_{x_{k}}h$ and $\nabla_{\xi}h$ into parallel to $\xi$ and perpendicular to $\xi$ components to obtain
\[
\begin{aligned}
T_{112} & = \int
\left[
\begin{array}[c]{l}
\lambda_1(\xi)\big(\mathbb{P}(\xi)\nabla_{\xi} \partial_{x_k} \partial_{k} h,\mathbb{P}(\xi) \nabla_{\xi} h\big)\\
	+		 \lambda_2(\xi)\big((I_3-\mathbb{P}(\xi))\nabla_{\xi} \partial_{x_k} \partial_{k} h,(I_3-\mathbb{P}(\xi)) \nabla_{\xi} h\big)
\end{array}
\right](1-\chi(|\xi|))\varrho\,m_0\\
& = -\int \left[
\begin{array}[c]{l}
(\partial_{k}\lambda_1)	\big(\mathbb{P}\nabla_{\xi} \partial_{x_k}  h,\mathbb{P} \nabla_{\xi} h\big)\\
+ (\partial_{k}\lambda_2) \big((I_3-\mathbb{P})\nabla_{\xi} \partial_{x_k} h,(I_3-\mathbb{P}) \nabla_{\xi} h\big)
\end{array}
\right](1-\chi(|\xi|))\varrho\,m_0\\
& \quad - \int \left[
\begin{array}[c]{l}
\lambda_1 \left[	\big((\partial_{k}\mathbb{P})\nabla_{\xi} \partial_{x_k}  h,\mathbb{P} \nabla_{\xi} h\big)+	\big(\mathbb{P}\nabla_{\xi} \partial_{x_k}  h,(\partial_{k}\mathbb{P}) \nabla_{\xi} h\big)  \right]\\
+ \lambda_2 \left[  \big(\partial_{k}(I_3-\mathbb{P})\nabla_{\xi} \partial_{x_k} h,(I_3-\mathbb{P}) \nabla_{\xi} h\big)+ \big((I_3-\mathbb{P})\nabla_{\xi} \partial_{x_k} h,\partial_{k}(I_3-\mathbb{P}) \nabla_{\xi} h\big)		\right]
\end{array}
\right](1-\chi(|\xi|))\varrho\,m_0\\
& \quad - \int \left[ \lambda_1\big(\mathbb{P}\nabla_{\xi} \partial_{x_k}  h,\mathbb{P} \nabla_{\xi}\partial_{k} h\big) +\lambda_2 \big((I_3-\mathbb{P})\nabla_{\xi} \partial_{x_k} h,(I_3-\mathbb{P}) \nabla_{\xi} \partial_{k}h\big)\right] (1-\chi(|\xi|))\varrho\,m_0\\
& \quad - \int \left[ \lambda_1\big(\mathbb{P}\nabla_{\xi} \partial_{x_k}  h,\mathbb{P} \nabla_{\xi} h\big) +\lambda_2 \big((I_3-\mathbb{P})\nabla_{\xi} \partial_{x_k} h,(I_3-\mathbb{P}) \nabla_{\xi} h\big)\right]\partial_{k}  \big((1-\chi(|\xi|))\varrho\,m_0\big)\\
& =: T_{1121}+ T_{1122}+ T_{1123}+ T_{1124}.
\end{aligned}
\]
By Lemma \ref{prop1}, $\left|\nabla_{\xi}\lambda_1\right|\lesssim \left<\xi\right>^{\gamma-1}$ and $\left|\nabla_{\xi}\lambda_2\right|\lesssim \left<\xi\right>^{\gamma+1}$, thus
\[
\left|T_{1121}\right|\lesssim \int \big(\left<\xi\right>^{\gamma-1} \left|\mathbb{P}\nabla_{\xi}\partial_{x_{k}}h\right| \left|\mathbb{P}\nabla_{\xi}h\right|+\left<\xi\right>^{\gamma+1} \left|(I_3-\mathbb{P})\nabla_{\xi}\partial_{x_{k}}h \right|\left|(I_3-\mathbb{P})\nabla_{\xi}h \right|	\big)\varrho\,m_0.
\]
Since $\left|\nabla_{\xi}\mathbb{P}\right|\lesssim \left<\xi\right>^{-1}$ when $|\xi|>1$,
\[
\begin{aligned}
\left|T_{1122}\right| & \lesssim \int \left< \xi \right>^{\gamma-1}\big[ \left|\nabla_{\xi}\nabla_{x}h\right|	\left|\mathbb{P}\nabla_{\xi}h\right| + \left|\mathbb{P}\nabla_{\xi} \partial_{x_{k}} h\right| \left|\nabla_{\xi} h\right| 	\big]\varrho\, m_0\\
&	\quad + \int  \left< \xi \right>^{\gamma+1}\big[ \left|\nabla_{\xi}\nabla_{x}h\right|	\left|(I_3-\mathbb{P})\nabla_{\xi}h\right| + \left|(I_3-\mathbb{P})\nabla_{\xi} \partial_{x_{k}} h\right| \left|\nabla_{\xi} h\right| 	\big]\varrho\, m_0.
\end{aligned}
\]
Furthermore, observing that
\[
T_{1123}=-T_{113},\quad T_{1124}= -T_{114},
\]
it then follows from combination of the above inequalities that
\[
\begin{aligned}
T_1
& = T_{11} + T_{12} = T_{111}+ T_{112} + T_{113} + T_{114} + T_{12}\\
& =  T_{111}+ T_{1121}+ T_{1122} +T_{12},
\end{aligned}
\]
and
\begin{equation}
\label{eq:reg-4}
\begin{aligned}
\left|T_1\right| & \lesssim \int  \left<\xi\right>^{\gamma+2}\left|\nabla_{x}h\right|\left|\nabla_{\xi} h\right|\varrho m_0 + \int \left<\xi\right>^\ga \left| \nabla_\xi h\right| \left|\nabla_{x}\nabla_\xi h\right|\varrho\,m_{0}\\
& \quad +\int  \left< \xi \right>^{\gamma+1}\big[ \left|\nabla_{\xi}\nabla_{x}h\right|	\left|(I_3-\mathbb{P})\nabla_{\xi}h\right| + \left|(I_3-\mathbb{P})\nabla_{\xi} \partial_{x_{k}} h\right| \left|\nabla_{\xi} h\right| 	\big]\varrho\, m_0.
\end{aligned}
\end{equation}

We still need to estimate $T_2$, use integration by parts to show
\[
\begin{aligned}
T_2
& =-\int h \left(\nabla_{\xi}\psi,\nabla_{x} h\right)\varrho\,m_0 = -\frac{1}{2} \int \left(\nabla_{\xi}\psi,\nabla_{x}h^2\right)\varrho\,m_0\\
&= \frac{1}{2} \int h^2 \left[\frac{(\nabla_{\xi}\psi,\nabla_{x}\varrho)}{\varrho}\right]\varrho\,m_0,
\end{aligned}
\]
thus
\begin{equation}\label{eq:reg-5}
\begin{aligned}
\left|T_2\right| & \lesssim
\begin{cases}
\frac{1}{D} \int h^2 \left<\xi\right>^{\gamma+1} \varrho\, m_0, & \gamma\in[-1,1],\\
\alpha\delta \int h^2 \left<\xi\right>^{2(1+\gamma)}\varrho\, m_0, & \gamma \in [-2,-1)
\end{cases}\\
& \lesssim \int h^2 \left<\xi\right>^{\gamma+1} \varrho\, m_0.
\end{aligned}
\end{equation}

Combining the equations \eqref{eq:reg-0}--\eqref{eq:reg-5}, we conclude
\[
\begin{aligned}
&\quad\frac{d}{dt} \int (\nabla_{x} h\, ,\nabla_{\xi} h)\varrho \,m_{0}\\
& \leq   -\int \left|\nabla_{x} h\right|^2 \varrho \, m_0 + C \int h^2 \left<\xi\right>^{\gamma+1} \varrho\, m_0 + C \int  \left<\xi\right>^{\gamma+2}\left|\nabla_{x}h\right|\left|\nabla_{\xi} h\right|\varrho m_0 \\
&\quad  + C \int \left<\xi\right>^{\gamma} \left|\mathbb{P}\nabla_{\xi}\partial_{k}h\right|\left|\mathbb{P}\nabla_{\xi}\partial_{x_k}h\right|\varrho\, m_0 + C\int \left<\xi\right>^{\gamma+2}\left|(I_3-\mathbb{P})\nabla_{\xi}\partial_{k}h\right|\left|(I_3-\mathbb{P})\nabla_{\xi}\partial_{x_k}h\right|\varrho\, m_0\\
&\quad + C \int \left<\xi\right>^\ga \left| \nabla_\xi h\right| \left|\nabla_{x}\nabla_\xi h\right|\varrho\,m_{0} \\
&\quad + C \int  \left< \xi \right>^{\gamma+1}\big[ \left|\nabla_{\xi}\nabla_{x}h\right|	\left|(I_3-\mathbb{P})\nabla_{\xi}h\right| + \left|(I_3-\mathbb{P})\nabla_{\xi} \partial_{x_{k}} h\right| \left|\nabla_{\xi} h\right| 	\big]\varrho\, m_0.
\end{aligned}
\]
Thanks to Young's inequality, we arrive at
\begin{equation}\label{eq:reg-6}
\begin{aligned}
\frac{d}{dt} \int (\nabla_{x} h\, ,\nabla_{\xi} h)\varrho \,m_{0}
&\leq -\int \left|\nabla_{x} h\right|^2 \varrho \, m_0 + C \left\|\left<\xi\right>^{\frac{\gamma+1}{2}} h\right\|^2_{L^2(\varrho\, m_0)}\\
& \quad + C\varepsilon_2 t\left\|\nabla_{x} h \right\|^2_{L^2_{\sigma}(\varrho\,m_0)}+  C(\varepsilon_2 t)^{-1}\left\|\nabla_{\xi} h \right\|^2_{L^2_{\sigma}(\varrho\,m_0)}.
\end{aligned}
\end{equation}

\noindent \underline{\it Step 3: Conclusion.}
Let us recall the estimates \eqref{eq:hypo-1}, we have $L^{2}$ estimate
\[
\frac{d}{dt} \int h^2 \varrho\, m_1 \leq -c_0 \left\|h \right\|_{L^2_\sigma(\varrho\,m_1)}^2,
\]
and $x$-derivative estimate,
\[
\frac{d}{dt} \int |\nabla_{x} h|^2 \varrho\,m_0 \leq -c_0 \left\| \nabla_x h \right\|^2_{L^2_\sigma(\varrho\,m_0)}.
\]
The fourth term in \eqref{eq_l6.1} is bounded by
\[
2\alpha_2 t \int (\nabla_{x} h\, ,\nabla_{\xi}  h)\varrho\,m_0\leq \alpha_2 \left(\varepsilon_3 t^2 \left\|\nabla_{x} h \right\|^2_{L^2(\varrho\,m_0)} + \varepsilon_3^{-1}  \left\|\nabla_\xi h \right\|^2_{L^2(\varrho\,m_0)}	\right).
\]
Together with $\xi$-derivative \eqref{eq_l6.6} and mixed term \eqref{eq:reg-6}, we have
\[
\begin{aligned}
\frac{d}{dt} \mathcal{F}(t,h(t)) & \lesssim \left\|h \right\|^2_{L^2_\sigma(\varrho\,m_1)}\Big(	-c_0 +\alpha_1 +\alpha_1 t +\alpha_1 \varepsilon_1^{-1} + \alpha_2 \varepsilon_3^{-1} +\alpha_2 t^2 \Big)\\
&\quad +  \left\|\nabla_{\xi}h \right\|^2_{L^2_\sigma(\varrho\,m_0)} t\,\Big(- c_0 \alpha_1 +\alpha_2 \varepsilon_2^{-1}		\Big)\\
&\quad +  \left\|\nabla_x h \right\|^2_{L^2(\varrho\,m_0)}t^2\,\Big(	-\alpha_2  +\alpha_1 \varepsilon_1+3\alpha_3 +\alpha_2\varepsilon_3		 \Big)\\
&\quad +  \left\|\nabla_x h \right\|^2_{L^2_\sigma(\varrho\,m_0)}t^3 \Big( - c_0 \alpha_3 +\alpha_2 \varepsilon_2 			\Big).
\end{aligned}
\]
Set $\varepsilon_2=\varepsilon_3=\varepsilon^2$, $\varepsilon_1=\varepsilon^3$, $\alpha_1=\varepsilon^{7/2}$, $\alpha_2=\varepsilon^6$, $\alpha_3=\varepsilon^{15/2}$. When $t\leq 1$, for sufficiently small $\varepsilon$, we immediately see all the coefficients above are non-positive. This implies the functional $\mathcal{F}$ is decreasing. Moreover, noting that $\alpha_1\alpha_3>2\alpha_2^2$, the mixed term in functional $\mathcal{F}$ can be dominated by other terms. Thus we conclude
\[
\left\| h(t) \right\|^2_{L^2(\varrho\,m_1)}+  \left( t \left\| \nabla_\xi h(t) \right\|^2_{L^2(\varrho\,m_0)}+t^3\left\|\nabla_{x} h(t) \right\|^2_{L^2(\varrho\,m_0)}\right) \lesssim\mathcal{F}(t,h(t))\leq\mathcal{F}(0,h(0))=\left\| h(0) \right\|^2_{L^2(\varrho\,m_1)}.
\]
This completes the proof of the lemma.
\end{proof}

\subsection{Proof of Lemma \ref{improve}}

Before going to the proof of Lemma \ref{improve}, we need to show that the operator $K$ is bounded in the space $H^{l}_{x}L^{2}_{\xi}(m_{0})$.
\begin{lemma}\label{bdK}
For $l\in \N$ and $g_{0}\in H^{l}_{x}L^{2}_{\xi}(\varrho \,m_{0})$, there exists a positive constant $C$ such that
\[
\left\|K g_{0}\right\|_{H^{l}_{x}L^{2}_{\xi}(\varrho\,m_{0})}\leq C\left\| g_{0}\right\|_{H^{l}_{x}L^{2}_{\xi}(\varrho\,m_{0})}.
\]
\end{lemma}
\begin{proof}
Since $K$ operator commutes with $x$-derivatives, it suffices to prove
\[
\left\|K g_{0}\right\|_{L^{2}(\varrho\,m_{0})}\leq C\left\|g_{0}\right\|_{L^{2}(\varrho\,m_{0})}.
\]
Let $u=g_0 \left(\varrho\,m_0\right)^{1/2}$, it is equivalent to show
\begin{equation}\label{eq:bdK-1}
\left\| K_\varrho u \right\|_{L^2}\leq C \left\| u\right\|_{L^2},
\end{equation}
where
\[
K_\varrho= \left(\varrho\,m_0\right)^{1/2} K \left(\varrho\,m_0\right)^{-1/2}.
\]
We discuss it for different $\ga$'s:

{\noindent \bf Case 1: $\ga\in\left[ -1,1 \right]$.} $\varrho$ is independent of $\xi$, so it commutes with $K$, one has
\[K_\varrho=m_0^{1/2}\widetilde{K}m_0^{-1/2}+\varpi\chi_R(\xi),\]
where
$$
\widetilde{K}g=\int_{\R^{3}}\mu^{-1/2}(\xi)\mu^{-1/2}(\xi_{*})(\nabla_{\xi},\nabla_{\xi_{*}}\cdot Z(\xi,\xi_{*}))g(\xi_{*})d\xi_{*}\,,
$$
and
$$
Z(\xi,\xi_{*})=\mu(\xi)\mu(\xi_{*})\Phi(\xi-\xi_{*})\,.
$$
Since $\chi_R(\xi)$ is a smooth cut-off function,  $\varpi\chi_R$ is a bounded operator. On the other hand, due to the function $\mu(\xi)$ exponentially decays in both $\xi$, applying Young's inequality  to the kernel function of integral operator, \eqref{eq:bdK-1} follows.

{\noindent \bf Case 2: $\ga\in \left[-2,-1\right)$.} $\varrho=e^{\alpha \vartheta(0,x,\xi)}$.
\[
\begin{aligned}
K_\varrho g & = e^{\alpha \vartheta(0,x,\xi)/2}\int_{\R^{3}} m_0^{1/2}(\xi) \tilde{k}(\xi,\xi_*)m_0^{-1/2}(\xi_*) e^{-\alpha \vartheta(0,x,\xi_*)/2} g(t,x,\xi_*) d\xi_*+ \varpi\chi_R(\xi )g(\xi)\,.
\end{aligned}
\]
We find the kernel function of $m_0^{1/2}K m_0^{-1/2}-K_\varrho$ is
\[
\tilde{k}(\xi,\xi_{*})m_0(\xi)^{1/2}m_0(\xi_{*})^{-1/2}\left(1-e^{\alpha\big(\vartheta(0,x,\xi)-\vartheta(0,x,\xi_{*} ) \big) }\right)\,.
\]
By similar argument as in Proposition \ref{weig_2} (see \eqref{K-Kw} and the paragraph above), we have
$$
\left\|(m_0^{1/2}K m_0^{-1/2}-K_\varrho)g\right\|_{L^{2}}\lesssim \alpha \left\|g\right\|_{L^{2}}\,,
$$
this together with  {\bf Case 1} imply the $L^{2}$ estimate.
%
\end{proof}

Now, it is ready to prove the regularization effect of the original linearized Landau equation in weighted space. Our strategy is to design a Picard-type iteration, which treats $Kf$ as a source term. Let $f$ be the solution to equation \eqref{prob}. The zeroth order approximation of the linearized Landau equation is
\begin{equation}\label{bot.3.b}
\left\{\begin{array}{l} \pa_{t}h^{(0)}=\mathcal{L}h^{(0)}\,
\\[4mm]
h^{(0)}(0,x, \xi)=f_{0}(x, \xi)\,.
\end{array}
\right.\end{equation}
Thus the difference $f-h^{(0)}$  satisfies
\begin{equation*}
\left\{\begin{array}{l} \pa_{t}(f-h^{(0)})=\mathcal{L}(f-h^{(0)})+K(f-h^{(0)})+Kh^{(0)}\,,
\\[4mm]
(f-h^{(0)})(0,x, \xi)=0\,.
\end{array}
\right.\end{equation*}
This motivates us to define the first order approximation $h^{(1)}$ by
\begin{equation}\label{bot.3.c}
\left\{\begin{array}{l} \pa_{t}h^{(1)}=\mathcal{L}h^{(1)}+Kh^{(0)}\,,
\\[4mm]
h^{(1)}(0,x, \xi)=0\,.
\end{array}
\right.\end{equation}
In general, we can define the $j^{\rm th}$ order approximation $h^{(j)}$, $j\geq 1$, as
\begin{equation}\label{bot.3.d}
\left\{\begin{array}{l} \pa_{t}h^{(j)}=\mathcal{L}h^{(j)}+Kh^{(j-1)}\,,
\\[4mm]
h^{(j)}(0,x, \xi)=0\,.
\end{array}
\right.\end{equation}

The singular wave part and the remainder part can be defined as follows:
$$
W^{(3)}=\sum_{j=0}^{3}h^{(j)}\,,\quad \mathcal{R}^{(3)}=u-W^{(3)}\,.
$$
Note that $\mathcal{R}^{(3)}$ solves the equation
\begin{equation}
\left\{\begin{array}{l} \pa_{t}\mathcal{R}^{(3)}+\xi\cdot\nabla_{x} \mathcal{R}^{(3)}=L\mathcal{R}^{(3)}+Kh^{(3)}\,,
\\[4mm]
\mathcal{R}^{(3)}(0,x, \xi)=0\,.
\end{array}
\right.\end{equation}

We divide our proof into some steps:
\newline\underline{\itshape Step 1: First derivative of $h^{(j)}$, $0\leq j\leq 3$ in small-time.} We want to show that for $0<t\leq1$,
\[
\left\|\nabla_{x}h^{(j)}\right\|_{L^{2}(\varrho\,m_{0})}\leq t^{(-3+2j)/2}\left\| f_{0}\right\|_{L^{2}(\varrho\,m_{1})}\,.
\]
The estimate of $h^{(0)}$ follows immediately from Lemma \ref{regul}. Note that
$$
h^{(1)}=\int_{0}^{t}e^{(t-s) \mathcal{L}} K  e^{s \mathcal{L}} f_{0}ds\,,
$$
hence
$$
\nabla_{x}h^{(1)}=\int_{0}^{t}\frac{(t-s)+s}{t}\nabla_{x}e^{(t-s) \mathcal{L}} K e^{s\mathcal{L}} u_{0}ds\,,
$$
we then have
\begin{align*}
\left\|\nabla_{x}h^{(1)}\right\|_{L^{2}(\varrho\,m_{0})}&\leq\int_{0}^{t}t^{-1}\left[(t-s)^{-1/2}+s^{-1/2}\right]ds\,
\|f_{0}\|_{L^{2}(\varrho\,m_{1})}\\
&\leq t^{-1/2}\|f_{0}\|_{L^{2}(\varrho\,m_{1})}\,.
\end{align*}

Similarly, note that
$$
h^{(2)}=\int_{0}^{t}\int_{0}^{s_{1}}e^{(t-s_{1}) \mathcal{L}} K e^{(s_{1}-s_{2}) \mathcal{L}} K e^{s_{2} \mathcal{L}} f_{0}ds_{2}ds_{1}\,,
$$
hence
$$
\nabla_{x}h^{(2)}=\int_{0}^{t}\int_{0}^{s_{1}}\frac{(s_{1}-s_{2})+s_{2}}{s_{1}}\nabla_{x}e^{(t-s_{1}) \mathcal{L}} K e^{(s_{1}-s_{2}) \mathcal{L}} K e^{s_{2} \mathcal{L}} f_{0}ds_{2}ds_{1}\,,
$$
we then have
\begin{align*}
\left\|\nabla_{x}h^{(2)}\right\|_{L^{2}(\varrho\,m_{0})}&
\leq\int_{0}^{t}\int_{0}^{s_{1}}s_{1}^{-1}\left[(s_{1}-s_{2})^{-1/2}+s_{2}^{-1/2}\right]ds_{2}ds_{1}\,
\|f_{0}\|_{L^{2}(\varrho\,m_{1})}\\
&\leq t^{1/2}\|f_{0}\|_{L^{2}(\varrho\,m_{1})}\,.
\end{align*}
The estimate of $h^{(3)}$ is similar and hence we omit the details.
\newline \underline{\itshape Step 2: Second $x$-derivative estimate of $h^{(j)}$,
$0\leq j\leq3$ in small time.} We want to show that for any $0<t\leq 1$,
\[
\left\|D_x^2 h^{(j)}\right\|_{L^2(\varrho\,m_0)}\leq C_j t^{-3+j}\left\|f_0\right\|_{L^2(\varrho\,m_2)}.
\]
Here we only show the detailed proof for $h^{(0)}$ and $h^{(1)}$ for brevity. The cases for $h^{(2)}$ and $h^{(3)}$ are similar. For any $0<t_{0}\leq1$ and $t_{0}/2<t\leq t_{0}$, we
have
\[
\nabla_{x}h^{(0)}(t)=e^{(t-t_{0}/2)\mathcal{L}}\nabla_{x}h^{(0)}(t_{0}/2)\,,
\]
hence (by Lemma \ref{regul})
\begin{equation}\label{h0-xx}
\left\Vert D_{x}^{2}h^{(0)}(t)\right\Vert _{L^{2}(\varrho\,m_0)}\lesssim
(t-t_{0}/2)^{-3/2}(t_{0}/2)^{-3/2}\Vert f_{0}\Vert_{L^{2}(\varrho\,m_2)}\,.
\end{equation}
If we take $t=t_{0}$, we have
\begin{equation}
\left\Vert D_{x}^{2}h^{(0)}(t_{0})\right\Vert _{L^{2}(\varrho\,m_0)}\lesssim
t_{0}^{-3}\Vert f_{0}\Vert_{L^{2}(\varrho\,m_2)}\,. \label{f0-xx}%
\end{equation}
Since $t_{0}\in(0,1)$ is arbitrary, this completes the estimate of $h^{(0)}$.
For $h^{(1)}$, let $0<t_{1}\leq1$ and $t_{1}/2<t\leq t_{1}$, then
\[
\nabla_{x}h^{(1)}(t)=e^{(t-t_{1}/2)\mathcal{L}}\nabla_{x}h^{(1)}(t_{1}%
/2)+\int_{t_{1}/2}^{t}e^{(t-s)\mathcal{L}}K\nabla_{x}h^{(0)}(s)ds\,,
\]
hence (by Lemma \ref{regul} and \eqref{h0-xx})
\[
\left\Vert D_{x}^{2}h^{(1)}(t)\right\Vert _{L^{2}(\varrho\,m_0)}\lesssim
(t-t_{1}/2)^{-3/2}(t_{1}/2)^{-1/2}\Vert f_{0}\Vert_{L^{2}(\varrho\,m_2)}+\int_{t_{1}%
	/2}^{t}s^{-3}\Vert f_{0}\Vert_{L^{2}(\varrho\,m_2)}ds\,.
\]
Now, take $t=t_{1}$, we get
\[
\left\Vert D_{x}^{2}h^{(1)}(t_{1})\right\Vert _{L^{2}(\varrho\,m_0)}\lesssim
t_{1}^{-2}\Vert f_{0}\Vert_{L^{2}(\varrho\,m_2)}\,.
\]
Since $t_{1}\in\left(  0,1\right)  $ is arbitrary, this completes the estimate
of $h^{(1)}$.

Combining above calculations, if we choose $m_{0}=1$, then we have that for $0<t<1$
$$
\left\|W^{(3)}(t)\right\|_{H^{2}_{x}L^{2}_{\xi}(\varrho) }\leq Ct^{-3}\left\|u_{0}\right\|_{L^{2}(\varrho\,\mathcal{M})}\,.
$$
Now, we only need to estimate the remainder part $\mathcal{R}^{(3)}$.
\newline \underline{\itshape Step 3: Second $x$-derivatives of $\mathcal{R}^{(3)}$ in small time.} Note that the solution of $\mathcal{R}^{(3)}$ is
$$
\mathcal{R}^{(3)}=\int_{0}^{t}\mathbb{G}^{t-s}Kh^{(3)}(s)ds\,,
$$
similar to the weighted energy estimate in Proposition \ref{weig_1} and Proposition \ref{weig_2} (in fact, it is easier), one can show that $\|\mathbb{G}^{t}\|_{H^{2}_{x}L^{2}_{\xi}(\varrho)}$ is uniform bounded for $0<t<1$, hence
$$
\left\|D_{x}^{2}\mathcal{R}^{(3)}\right\|_{L^{2}(\varrho)}\leq\int_{0}^{t}\left\|D_{x}^{2}h^{(3)}\right\|_{L^{2}(\varrho)}(s)ds\leq Ct\|f_{0}\|_{L^{2}(\varrho\,\mathcal{M})}\,.
$$
In conclusion, we have that for $0<t\leq1$
$$
\left\|f(t)\right\|_{H^{2}_{x}L^{2}_{\xi}(\varrho)}\leq Ct^{-3}\left\|f_{0}\right\|_{L^{2}(\varrho\,\mathcal{M})}\,.
$$


\end{document}